\documentclass[journal]{IEEEtran}

\usepackage{booktabs} % For professional looking tables
\usepackage{multirow}

\usepackage{enumitem}

\usepackage{smartdiagram}
\usesmartdiagramlibrary{additions}

\graphicspath{{Figs/}}

\usepackage{pdfpages}
\usepackage{pdflscape}
\usepackage{amsmath}
\usepackage{amsthm}

\DeclareMathOperator*{\argmin}{arg\,min}

\DeclareMathOperator*{\diag}{diag}
\DeclareMathOperator*{\col}{col}
\DeclareMathOperator*{\row}{row}
\DeclareMathOperator*{\vect}{vec}
\DeclareMathOperator*{\st}{s.t.}

\usepackage{amsfonts}

\usepackage{color}

\usepackage{algorithm}
\usepackage{algorithmic}
\let\Algorithm\algorithm

\renewcommand\algorithm[1][]{\Algorithm[#1]\setstretch{1}}

\newcommand{\LINEFOR}[2]{%
    \STATE\algorithmicfor\ {#1}\ \algorithmicdo\ {#2} \algorithmicend\ \algorithmicfor%
}

\usepackage{setspace}
\usepackage{enumitem}
\usepackage{xcolor}

\usepackage{bm}
\usepackage{bbm}

\DeclareMathOperator{\Tr}{Tr}
\DeclareMathOperator{\rank}{rank}

\usepackage{tikz}
\usepackage{pgf}
\usetikzlibrary{arrows,automata}
\usetikzlibrary{shapes}
\tikzstyle{data44}=[rectangle split,rectangle split parts=2,draw,text centered]

\usepackage{subfigure}

\DeclareMathAlphabet{\mathcal}{OMS}{cmsy}{m}{n}

\usetikzlibrary{matrix}

\tikzset{
  BarreStyle/.style =   {opacity=.3,line width=14 mm,color=#1},
  node style ge/.style={},
  node style sp/.style={},
  yl/.style={},
  arrow style mul/.style={},
}

\newtheorem{approximation}{Approximation}

\newtheorem{proposition}{Proposition}
\newtheorem{corollary}{Corollary}
\newtheorem{remark}{Remark}

\usepackage{threeparttable,booktabs}

\usepackage[colorlinks,
            linkcolor=blue,
            anchorcolor=blue,
            citecolor=blue
            ]{hyperref}

\usepackage{cite}

\usepackage{mdwlist}

\newcommand{\hlr}[1]{\textcolor{black}{#1}}

\allowdisplaybreaks

\begin{document}
\title{Dispatch of Virtual Inertia and Damping: Numerical Method with SDP and ADMM}

\author{Tong Han,~\IEEEmembership{}
        David J. Hill~\IEEEmembership{}
\thanks{This work was supported by the Research Grants Council of the Hong Kong Special Administrative Region through the General Research Fund under Project No. 17209419.}
\thanks{T. Han and D. J. Hill are with the Department of Electrical and Electronic Engineering, University of Hong Kong, Hong Kong (e-mail: hantong@eee.hku.hk, dhill@eee.hku.hk).}% <-this % stops a space
% \thanks{D. J. Hill is with the Department of Electrical and Electronic Engineering, University of Hong Kong, Hong Kong, and also with the School of Electrical and Information Engineering, University of Sydney, Sydney, NSW 2006, Australia (e-mail: dhill@eee.hku.hk).}
\vspace{-20pt}
}

\markboth{In Press, International Journal of Electrical Power \& Energy Systems}%
{Shell \MakeLowercase{\textit{et al.}}: Bare Demo of IEEEtran.cls for IEEE Journals}

\maketitle

\begin{abstract}
Power grids are evolving toward 100\% renewable energy interfaced by inverters. Virtual inertia and damping provided by inverters are essential to synchronism and frequency stability of future power grids. This paper numerically addresses the problem of dispatch of virtual inertia and damping (DID) among inverters in the transmission network. The DID problem is first formulated as a nonlinear program (NLP) by the Radua collocation method which is flexible to handle various types of disturbances and bounds constraints. Since the NLP of DID is highly non-convex, semi-definite programming (SDP) relaxation for the NLP is further derived to tackle the non-convexity, followed by its sparsity being exploited hierarchically based on chordality of graphs to seek enhancement of computational efficiency. Considering high dimension and inexactness of the SDP relaxation, a feasibility-embedded distributed approach is finally proposed under the framework of alternating direction method of multipliers (ADMM), which achieves parallel computing and solution feasibility regarding the original NLP. Numerical simulations carried out for five test power systems demonstrate the proposed method and necessity of DID.

\end{abstract}


\IEEEpeerreviewmaketitle

\section*{\hlr{Nomenclature}}

\subsection{Scalers}
\addcontentsline{toc}{section}{Nomenclature}
\begin{IEEEdescription}[\IEEEusemathlabelsep\IEEEsetlabelwidth{$V_1,V$}]
\item[$\!\!\!\!\!\!$$V_i$] Voltage magnitude of bus $i$.
\item[$\!\!\!\!\!\!$$p_i$] Mechanical power input for $i \in \mathcal{N}_s$, power setpoint for $i \in \mathcal{N}_v$, negative of load power independent of frequency for $i \in \mathcal{N}_l$ and 0 for $i \in \mathcal{N}_o$.
\item[$\!\!\!\!\!\!$$b_{j_1, j_2}$] Susceptance of branch $(j_1, j_2)$.
\item[$\!\!\!\!\!\!$$m_i$, $d_i$] Inertia/damping coefficient of generator $i$.
\item[$\!\!\!\!\!\!$$d_{li}$] Frequency coefficient of load $i$.
\item[$\!\!\!\!\!\!$$\underline{b}_i$] Equivalent short-circuit susceptance when ignoring the short-circuit resistance.
\item[$\!\!\!\!\!\!$$J_k^i$] Component of objective function $J$ corresponding to disturbance $k$ and time element $i$. 
\item[$\!\!\!\!\!\!$$\hat{J}$] Approximation of $J$ in (P1).
\item[$\!\!\!\!\!\!$$\beta_{cf}$]  A proper large number to guarantee positive definiteness of $A_{adj} \!+\! \beta_{cf} I$.  
\item[$\!\!\!\!\!\!$$\tilde{\rho} > 0$ ] The penalty parameter.
\item[$\!\!\!\!\!\!$$\sigma_{\mathcal{C}_{ej}}^1$] The 1-th singular value of matrix $\hat{\bm{Z}}_{\mathcal{C}_{ej}}^{(\kappa + 1)}  + \frac{1}{\tilde{\rho}}  \hat{\bm{\Lambda}}_{\mathcal{C}_{ej}}^{(\kappa)}$.  
\item[$\!\!\!\!\!\!$$r^{\kappa}$, $s^{\kappa}$] Primal and dual residuals at iteration $\kappa$, respectively
\item[$\!\!\!\!\!\!$$\epsilon^{\mathrm{abs}}, \epsilon^{\mathrm{rel}}\!$] Absolute tolerance and relative tolerance.
\item[$\!\!\!\!\!\!$$\varphi_s, \tilde{\varphi}_s$ ]  Slack variables. 
\item[$\!\!\!\!\!\!$$\bm{Z}_{j,(1,2)}^{\text{md}}$] The entry in the 1-th row and 2-th column of $\bm{Z}_{j}^{\text{md}}$.  
\end{IEEEdescription}

\subsection{Vectors}
\addcontentsline{toc}{section}{Nomenclature}
\begin{IEEEdescription}[\IEEEusemathlabelsep\IEEEsetlabelwidth{$V_1,V$}]
\item[$\!\!\!\!\!\!$$\theta$] Phase angle of all buses.
\item[$\!\!\!\!\!\!$$\omega$] Angular frequency of all generators.
\item[$\!\!\!\!\!\!$$\omega_{t_0}$, $\theta_{t_0}$] Initial values of $\omega^k$ and $\theta^k$.
\item[$\!\!\!\!\!\!$$\underline{\omega}^k$, $\overline{\omega}^k$] Lower/upper frequency bound.
\item[$\!\!\!\!\!\!$$\overline{\delta}$] Upper bound of angle differences.
\item[$\!\!\!\!\!\!$$p_k$] $\text{col}(p_i^k) \in \mathbb{R}^{n_a}$ with $p_i^k$ being $p_i$ for disturbance $k$.
\item[$\!\!\!\!\!\!$$\underline{p}_g$, $\overline{p}_g$] Lower/upper active power limit of generators.
\item[$\!\!\!\!\!\!$$\underline{m}$, $\overline{m}$] Lower/upper bound of inertia coefficients.
\item[$\!\!\!\!\!\!$$\underline{d}$, $\overline{d}$] Lower/upper bound of damping coefficients.
\item[$\!\!\!\!\!\!$$\bm{\theta}_i^{kT}\!\!$, $\!\bm{\omega}_i^{kT}$] $(\bm{\theta}_i^{k})^T$, $(\bm{\omega}_i^{k})^T$.
\item[$\!\!\!\!\!\!$$\underline{\omega}^k_{ir}$, $\overline{\omega}^k_{ir}$] $\underline{\omega}^k(t_{i-1} + \tau_r h_i^k)$, $\overline{\omega}^k(t_{i-1} + \tau_r h_i^k)$.
\item[$\!\!\!\!\!\!$$p^k_{ir}$, $\tilde{p}_{ir}^k$] $p^k(t_{i-1} + \tau_r h_i^k)$, $\tilde{p}^k(t_{i-1} + \tau_r h_i^k)$.
\item[$\!\!\!\!\!\!$$\bm{\theta}$, $\bm{\omega}$] $\col(\bm{\theta}_i^k|_{k \in \mathcal{D}, i \in \mathbb{T}^k})$, $\col(\bm{\omega}_i^k|_{k \in \mathcal{D}, i \in \mathbb{T}^k})$.
\item[$\!\!\!\!\!\!$$\theta_{0r}^k$, $\omega_{0r}^k$] Constants equal to $\theta_{t_0}$ and $\omega_{t_0}$.
\item[$\!\!\!\!\!\!$$\bm{x}$] $\text{col}(M \mathbbm{1}, D \mathbbm{1},\cdots\!, \bm{l}_{\text{d} (i-1)}^k, \bm{\theta}_i^k, \bm{\omega}_i^k, \bm{l}_{\text{m} i}^k, \bm{l}_{\text{d} i}^k, \bm{\theta}_{i+1}^k, $ $\cdots |_{k \in\! \mathcal{D}, i\in \mathbb{T}^k} )$.
\item[$\!\!\!\!\!\!${$[\bm{x}]_i^k$}] Sub-vector of $\bm{x}$ related to disturbance $k$ and time element $i$, i.e., $[\bm{x}]_i^k \!=\! \text{col}(\!M \mathbbm{1}, D \mathbbm{1}, \bm{\theta}_i^k, \bm{\omega}_i^k, \bm{l}_{\text{m} i}^k, \bm{l}_{\text{d} i}^k)$. 
\item[$\!\!\!\!\!\!$$\bm{\alpha}, \bm{\beta}, \bm{\varsigma}$] $\col(\bm{\alpha}_{(r,\imath)i}^k)$, $\col(\bm{\beta}_{(r,\imath)i}^k)$, $\col(\bm{\varsigma}_{(r,\imath)i}^k)$, with $k \in \mathcal{D}$, $i \in \mathbbm{T}^k$, $r \in \{0,...,n_c\}$, $\imath \in \mathcal{B}$, $\bm{\alpha}_{(r,\imath)i}^k \in \mathbb{R}^4$, $\bm{\beta}_{(r,\imath)i}^k \in \mathbb{R}^3$ and $\bm{\varsigma}_{(r,\imath)i}^k \in \mathbb{R}^3$. 
\item[$\!\!\!\!\!\!$$\diag(\!\zeta\!)_i$] Arbitrary disjoint sub-vector of $\diag(\!\zeta\!)$ satisfying $[\diag(\zeta)_1^T\!,\! \diag(\zeta)_2^T,\!...,\!\diag(\zeta)_{n_\mathrm{d}}^T] \!\!\!=\!\! \diag(\zeta)^T$.
\item[$\!\!\!\!\!\!${$\!\mathrm{upper}(\!\zeta\!)_i$}] Analogous to $\diag(\zeta)_i$.
\item[$\!\!\!\!\!\!$$u_{\mathcal{C}_{\!ej}}^1\!, {v_{\mathcal{C}_{\!ej}}^{1}}\!$] The left and right singular vectors corresponding to $\sigma_{\mathcal{C}_{ej}}^1$, respectively.   
\end{IEEEdescription}

\subsection{Matrices}
\addcontentsline{toc}{section}{Nomenclature}
\begin{IEEEdescription}[\IEEEusemathlabelsep\IEEEsetlabelwidth{$V_1,V$}]
\item[$\!\!\!\!\!\!$$M\!$, $\!\!D\!$, $\!\!D_l$] $\text{diag}(m_i|_{i \in \mathcal{N}_g})$, $\text{diag}(d_i|_{i \in \mathcal{N}_g})$, $\text{diag}(d_{li}|_{i \in \mathcal{N}_l})$.
\item[$\!\!\!\!\!\!$$B$] $\text{diag}(V_{j_1}V_{j_2}b_{j_1, j_2}|_{(j_1, j_2) \in \mathcal{B}})$.
\item[$\!\!\!\!\!\!$$E_g$, $E_l$] Incidence matrices showing the relationship between $\mathcal{N}_g$ and $\mathcal{N}$, and $\mathcal{N}_l$ and $\mathcal{N}$.
\item[$\!\!\!\!\!\!$$E_o$, $E_n$] Incidence matrices showing the relationship between $\mathcal{N}_o$ and $\mathcal{N}$, and $\mathcal{N}$ and $\mathcal{B}$.
\item[$\!\!\!\!\!\!$$\Omega_i^k$, $\Theta_i^k$] Collocation coefficient matrices for profile of $\omega^k$ and $\theta^k$ at time element $i$.
\item[$\!\!\!\!\!\!$$E_{gl}$,$\check{M}$,$\check{D}$] $\text{col}(E_g, E_l)$, $\text{diag}(M,...,M)$, $\text{diag}(D,...,D)$.
\item[$\!\!\!\!\!\!$$B_{ir}^k$] $B^k(t_{i-1} + \tau_r h_i^k)$.
\item[$\!\!\!\!\!\!$$\bm{\Lambda}_{0i}^k$, $\bm{\Lambda}_{1i}^k$] $\diag\left( \cos ( \bm{A}_{1i}^k \bm{\theta}_0^k) \right)$, $\sin ( \bm{A}_{1i}^k \bm{\theta}_0^k) - \bm{\Lambda}_{0i}^k \bm{\theta}_0^k$.
\item[$\!\!\!\!\!\!$$\bm{\ell}_{\omega j}( \tau_r)$] Equal to $\bm{\ell}_{\omega}( \tau_r) )$ with only the row corresponding to $d_j$ remained and others replaced by 0.
\item[$\!\!\!\!\!\!$$\bm{O}_{(r,j)}^1$] Matrix in $\mathbb{R}^{n_g \times n_g (n_c + 1)}$ with the element corresponding to $d_j$ and angular speed of generator $j$ in $\bm{\omega}_{ir}^k$ being 1 and other being 0.
\item[$\!\!\!\!\!\!${$[\bm{X}]_i^k$}] Principal submatrix of $\bm{X}$ related to disturbance $k$ and time element $i$, given as $[\bm{x}]_i^{k} [\bm{x}]_i^{kT}$.
\item[$\!\!\!\!\!\!$$\tilde{\bm{Q}}_i^{k}$]  Sub-matrix of $\bm{Q}_i^{k}$ by removing the last 2 block rows.
\item[$\!\!\!\!\!\!$$\tilde{\bm{P}}_1 $, $\bm{A}_{8i}^k$]  $\frac{1}{2}[[O, I],[I, O]]$, $\diag(\col(1,0,-1))$.
\item[$\!\!\!\!\!\!$$\tilde{\bm{A}}_i^k, \tilde{\bm{b}}_i^k$]  Sub-matrices of ${\bm{A}}_i^k$ and ${\bm{b}}_i^k$, by removing the 3th and 4th block rows, respectively.
\item[$\!\!\!\!\!\!$$\bm{P}_{(\tilde{r}, \imath)i}^k$] With the same block structure as $\bm{P}_{(\hbar, r,j) i}^k$, which in block 3-tuple form, is given by $(\bm{\theta}_i^k, \bm{\theta}_i^k,  \frac{1}{2}\vartheta \bm{A}_{1i}^{k T} O_{(r,j)}^1 \bm{A}_{1i}^{k})$.
\item[$\!\!\!\!\!\!$$O_{(\tilde{r},\imath)}^1$]  Matrix in $\mathbb{R}^{n_b(n_c + 1) \times n_b(n_c + 1)}$ with diagonal elements corresponding to the $r$-th time element and branch $\imath$ being 1 and others being 0.
\item[$\!\!\!\!\!\!$$\bm{A}_{9i}^k$, $\!\!\bm{A}_{10i}^k$]  $\diag(\col(\pi \vartheta, \frac{\sin \theta_b}{\theta_b}, \pi \vartheta)) \bm{A}_{1i}$, $\diag(\theta_{cb}^T)$.
\item[$\!\!\!\!\!\!$$\tilde{\bm{P}}_2$] $[[I, O],[O, O]]$.
\item[$\!\!\!\!\!\!$$\mathcal{S}_{\mathcal{C}}(\bm{Z})$]  Principal submatrix of $\bm{Z}$ defined by the index set $\mathcal{C}$ or a new defined matrix for the principal submatrix. 
\item[$\!\!\!\!\!\!$$A_{adj}$]  Adjacent matrix of $\mathcal{G}_n$. 

\item[$\!\!\!\!\!\!$$\bm{P}_{\mathcal{C}_{ej}}$] Matrix satisfying $ \sum_{ \mathcal{C}_{ej} \!\in\! \mathcal{K}_i^k \cup \tilde{\mathcal{K}}_i^k } \!\!\Tr(\!\bm{P}_{\!\mathcal{C}_{ej}} \! \hat{\bm{Z}}_{\!\mathcal{C}_{ej}}\!) \!\!\!=\!\!\! \Tr(\!\bm{P}_{0i}^k [\!\bm{X}\!]_i^k\!)$.

\item[$\!\!\!\!\!\!$$\bm{P}_{\mathcal{C}_{ej}}^{\gamma}$] Matrix to make constraints (\ref{eq-sparsity-7:2}) equivalent to (\ref{eq-sdp-1:2}), (\ref{eq-sdp-1:4}), (\ref{eq-sdp-4:2}) to the first equation of (\ref{eq-sdp-4:4}) and (\ref{eq-sdp-4:5-1}) to (\ref{eq-sdp-4:5-4}). 
\item[$\!\!\!\!\!\!$$\mathcal{O}_s$, $\!\mathcal{O}'_s$] Properly sized matrix vectors consisting of zero matrices. 
\item[$\!\!\!\!\!\!$$\mathcal{Z}_\text{a}$] The matrix vector consisting of all auxiliary matrix variables. 
\item[$\!\!\!\!\!\!$$\hat{\mathcal{Z}}_s$] $ [\hat{\bm{Z}}_{\mathcal{C}_{ej}}]^T$ with $\mathcal{C}_{ej} \in \mathcal{K}_i^k$ and $(k,i) \in \Xi_s $ 
\end{IEEEdescription}

\subsection{Sets}
\addcontentsline{toc}{section}{Nomenclature}
\begin{IEEEdescription}[\IEEEusemathlabelsep\IEEEsetlabelwidth{$V_1,V$}]
\item[$\!\!\!\!\!\!$$\mathcal{B}$] Terminal bus pairs of all branches, and $|\mathcal{B}| = n_b$.

\item[$\!\!\!\!\!\!$$\mathbb{S}^{|\mathcal{V}|}_{\mathcal{E}}$]  The set of symmetric $|\mathcal{V}| \times |\mathcal{V}|$ matrices with only entries in $\mathcal{E}$ specified.
    
\item[$\!\!\!\!\!\!$$\mathcal{I}(\cdot, \cdot)$] The set of row (or column) index of $\bm{Z}$ corresponding to that in parentheses. For example, $\mathcal{I}(\bm{\theta}_i^k, \cdot)$ denotes the set of row index of $\bm{Z}$ corresponding to $\bm{\theta}_i^k$, $\mathcal{I}(\bm{\theta}_i^k, j)$ denotes the set of row index of $\bm{Z}$ corresponding to $\bm{\theta}_i^k$ and bus (or generator) $j$ or collocation point $j$, $\mathcal{I}(-1)$ represents the index of the last row of $\bm{Z}$, and  $\mathcal{I}(k, i)$ represents the set of row index of $\bm{Z}$ only corresponding to disturbance $k$ and time element $i$,i.e., $\mathcal{I}(k, i) = \mathcal{I}(\bm{\theta}_i^k, \cdot) \cup \mathcal{I}(\bm{\omega}_i^k, \cdot) \cup \mathcal{I}(\bm{l}_{\text{m}i}^k, \cdot) \cup \mathcal{I}(\bm{l}_{\text{d}i}^k, \cdot)$.

\item[$\!\!\!\!\!\!$$\text{ch}( \mathcal{C}_{ej} )$] Children of clique $\mathcal{C}_{ej}$ in clique tree $\mathcal{T}$.
\item[$\!\!\!\!\!\!$$\mathcal{K}^j$] $\{\mathcal{C}_{ej}|\mathcal{C}_{ej} \!\in\! \bigcup_{k \in \mathcal{D}, i\in \mathbbm{T}^k} \mathcal{K}_{e1}^{ki}\} $ with $j \!\in\! \mathcal{N}_g$.
\item[$\!\!\!\!\!\!$$I_j, I'_j$] Set of indices corresponding to the intersection between cliques $\mathcal{C}_{ej}$ and $\mathcal{C}'_{ej}$, for $\hat{\bm{Z}}_{\mathcal{C}_{ej}}$ and $\hat{\bm{Z}}_{\mathcal{C}'_{ej}}$, respectively.
\item[$\!\!\!\!\!\!$$I_\omega$, $I'_\omega$] Sets of indices consisting of $-1$ and the index corresponding to $\omega_{i0}^k$ of generator $j$, $-1$ and the index corresponding to $\omega_{i-1,n_c}^k$ of generator $j$. 
\item[$\!\!\!\!\!\!$$I_\theta$, $\!I'_\theta$] Sets of indices consisting of $-1$ and the index corresponding to $\theta_{i0}^k$ of buses $\mathcal{C}_{nj}$, $-1$ and the index corresponding to $\!\theta_{i-1,n_c}^k\!\!$ of buses $\!\mathcal{C}_{nj}$. 
\item[$\!\!\!\!\!\!\!$$\max^\circ\!\!$, $\!\min^{\!\circ}$] Maximum/minimum values of each continuous part of an integer set, e.g., $\text{max}^\circ \{1,2,4,5,9\} \!\!=\!\! \{2,5,9\}$ and $\text{min}^\circ \{1,2,4,5,9\} \!\!=\!\! \{1,5,9\}$. 
\item[$\!\!\!\!\!\!$$\Im(\Xi_s)$] An arbitrary subset of $\Xi_s$ with one element.
\item[$\!\!\!\!\!\!$$\tilde{\mathcal{M}}$, $\underline{\mathcal{M}}$]$ \bigcup_{s \in \mathbbm{P}} \tilde{\mathcal{M}}_s$, $\bigcup_{s \in \mathbb{P}} \underline{\mathcal{M}}_s$.
\item[$\!\!\!\!\!\!$$\mathbb{Y}_s$] The feasible region of $\mathcal{Y}_s$, i.e., $\{ \mathcal{Y}_s | \rank( \mathcal{Y}_s)  \leq 1 \}$.
\item[$\!\!\!\!\!\!$$\mathbb{Z}_{s}$]  The feasible region defined by inner constraints with $(k,i) \!\in\! \Xi_s$ and coupling constraints $(k,i) \!\in\! \Xi_s$ satisfying $(k,i-1) \in \Xi_s$ or $i=1$.
\item[$\!\!\!\!\!\!$$\mathbb{R}_\text{a}$] The feasible region of $\mathcal{Z}_\text{a}$. 
\end{IEEEdescription}

\vspace{10pt}

\section{Introduction}

\IEEEPARstart{F}{uture} power grids will feature high penetration of renewable energy to mitigate global climate change. Inverter-interfaced generation will increasingly substitute for conventional synchronous generators, which however, can cause critical synchronism and frequency stability problems due to low, time-varying and heterogeneous system inertia \cite{4-999, 4-999-90}. Inertia emulation and droop control have been proposed as effective remedies to improve the synchronism and frequency performance of future power grids \cite{4-999-128, 4-999-182}.

Differently to synchronous generators whose inertia and damping are inherent physical properties, virtual inertia and damping of inverters are control parameters of control loops and thus are tunable. Following conventional power grids, inertia and damping coefficients would be set as constants. However, due to high heterogeneity in operating conditions of future power grids causing by wide variations on both generation and demand sides, a fixed setting of virtual inertia and damping is potentially unable to guarantee synchronism and frequency stability under all operating conditions. More importantly, system performances with the fixed setting would probably be far from the optimal with inertia and damping being different for each operating condition.

%Variations in distribution of generation and load can be equivalently regarded as variations in distribution  of inertia and damping, and thus leading to degradation in performances.
%The degradation in performances from ill-conceived inertia distribution has been verified \cite{4-999-208, 4-522}.

Naturally inspired by economic dispatch in power grids, we define the concept of \textit{dispatch of virtual inertia and damping} (DID): short-term determination of the optimal distribution of virtual inertia and damping among a number of inverters, to improve the synchronism and frequency performance, at the lowest control efforts, subject to operational constraints. Compared to economic dispatch with the aim to determine the steady-state optimal active power generation in terms of economic (generation cost), the aim of DID is to determine the optimal active power generation during dynamic processes mainly regarding security (synchronism and frequency stability). It is pointed out that some novel control mechanisms for the fast dynamics have been also developed as an improved alternative to the virtual inertia solution, such as the dynam-i-c droop control \cite{4-630, 4-994}. Nevertheless, while there is still substantial legacy generation and controls, the DID process will be useful for years to come.

Although the DID problem has not been explicitly proposed previously, researchers have explored some analogues of DID recently \cite{4-521, 4-520,  4-138, 4-239, 4-522,  4-236, 4-999-208,  4-237, 4-519}. These studies, known as placement or design of virtual inertia and primary control, mainly focus on optimization of virtual inertia and damping from a planning perspective. Whether from the operating or planning perspective, essentially the problem is the same, that is, an optimal control problem for time-invariant system parameters, i.e., inertia coefficients and damping coefficients. Five aspects are of concern, namely disturbances, system dynamic models, performance metrics, constraints and solution approaches. For the first aspect, power-step disturbances \cite{4-520, 4-138, 4-239, 4-522, 4-236, 4-237}, power-impulse disturbances \cite{4-138, 4-999-208} and unit-variance stochastic white noise power disturbances \cite{4-999-208} were considered to model faults such as generator outages, fluctuations in renewable generation and load steps. For system dynamic models, all existing work applied linear dynamic system models by linearization \cite{ 4-520, 4-138, 4-236, 4-237} or utilizing DC power flow \cite{4-522, 4-999-208}. Inverter dynamics were modelled with different fidelity, from the system-level \cite{4-520, 4-239, 4-999-208, 4-237} to the device-level \cite{4-138, 4-522,4-236}. Dependency of loads on frequency was considered in most work \cite{4-239, 4-522, 4-236, 4-999-208, 4-237}. For measuring system performance, different metrics or their combination were used, including the damping ratio \cite{4-236, 4-237}, the frequency overshoot \cite{4-236, 4-237}, the frequency nadir \cite{4-138, 4-522}, the Rate of Change of Frequency (RoCoF) \cite{4-138, 4-522, 4-236}, and $\mathcal{H}_2$ norms of linear systems with time-integrated quadratic forms in the voltage angle, frequency deviation, RoCoF or control efforts \cite{4-520, 4-138, 4-522, 4-999-208}. Regarding constraints, approximate inverter power limits \cite{4-138,4-522, 4-236}, inertia and damping coefficient bounds \cite{4-138, 4-239, 4-522, 4-999-208, 4-237} and constraints on some of the performance metrics \cite{4-236, 4-237} were considered. The above four aspects determine the problem formulation which is generally a non-convex and large-scale optimization model. Gradient-based optimization methods \cite{4-138, 4-239, 4-522, 4-999-208} and sequential linear programming methods \cite{4-236, 4-237} were applied to find a local optimum. Analytical solutions of the optimization model were also derived under restrictive assumptions \cite{4-999-208, 4-520}.

Existing studies are incomplete in two aspects. Firstly, the linear power flow model and time-invariant bound constraints were adopted in most previous formulations. Therefore, these formulations are incapable of handling some practical issues, e.g., time-varying frequency bounds and system performance under large disturbances. Secondly, most solution approaches neglected the non-convexity of optimization models, which probably produce inferior optimal solutions regarding global optimality. In light of the above incompleteness, this paper addresses the DID problem from a numerical perspective, aiming to develop a flexible DID formulation and a computationally efficient solution method to obtain near-globally optimal solutions. Specifically, the main contributions of this paper are fourfold: \textit{(\romannumeral1)} a nonlinear programming (NLP) formulation of DID is developed by the Radua collocation method, where various forms of disturbances and constraints can be easily taken into account; \textit{(\romannumeral2)} semi-definite programming (SDP) relaxation for the NLP of DID is derived to address non-convexity of the NLP model; \textit{(\romannumeral3)} to improve computational efficiency, sparsity in the SDP relaxation is exploited hierarchically based on the clique decomposition approach and the proposed proposition regarding chordality of graphs; and \textit{(\romannumeral4)} under the framework of alternating direction method of multipliers (ADMM), a feasibility-embedded distributed approach is proposed to solve the SDP relaxation for DID in parallel and with solution feasibility to the original NLP being guaranteed. 

\begin{figure}[h]
    \centering
    \includegraphics[scale=1.12]{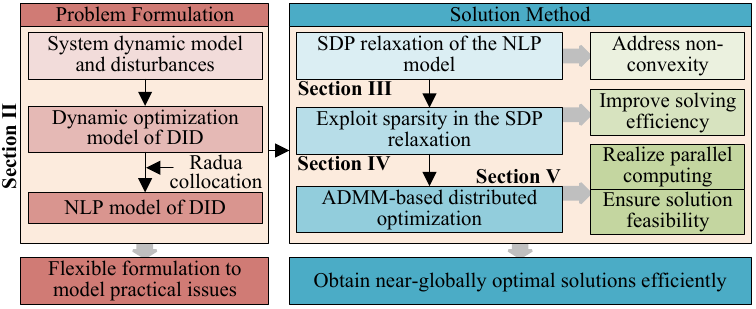}
    \caption{\hlr{The organization and purpose of Section \ref{section-1} to Section \ref{section-4}}}
    \label{fig-2-1-1}
\end{figure}

\hlr{The rest of this paper contains Section$\!$ \ref{section-1} to Section$\!$ \ref{section-4} $\!$organized as shown in Fig.$\!$ \ref{fig-2-1-1} which also gives the purpose of each section, Section$\!$ \ref{section-5} giving the results of DID of five test systems with our proposed model and solution approach, and Section$\!$ \ref{section-6} $\!$making a conclusion and a prospect for future works.}

\section{Problem Formulation}\label{section-1}

In this section, system dynamic models and types of disturbances adopted in the DID problem are first introduced. Then, the dynamic optimization model of the DID problem is formulated, which is further transcribed into a tractable finite-dimensional NLP by the Radua collocation method.

\subsection{System Dynamic Model and Disturbances}\label{sec-6-dynamic-model}

Consider a transmission network that consists of $n_g$ generation buses including $n_s$ synchronous generators and $n_v$ grid-forming inverters, $n_l$ load buses and $n_o$ buses with neither generator nor load. Denote by $\mathcal{N}_g$, $\mathcal{N}_s$, $\mathcal{N}_v$, $\mathcal{N}_l$ and $\mathcal{N}_o$ the sets of these five types of buses, respectively. The set of all the buses is denoted by $\mathcal{N}$ with $|\mathcal{N}|= n_a$. We further divide $\mathcal{N}_v$ into $\mathcal{N}_{v_o}$, $\mathcal{N}_{v_d}$, $\mathcal{N}_{v_m}\!$ and $\!\mathcal{N}_{v_{dm}}$, which contain $n_{v_o}$ inverters with fixed damping and inertia, $n_{v_d}$ inverters with adjustable damping but fixed inertia, $n_{v_m}$ inverters with adjustable inertia but fixed damping, and $n_{v_{dm}}$ inverters with adjustable damping and inertia, respectively.

\textit{Dynamic Model}. The structure-preserving model is widely employed for dynamic analysis of transmission networks with only synchronous generators and frequency-dependent loads \cite{4-999-102}. By further embedding grid-forming inverter dynamics which are essentially the same as that of synchronous generators in the context of this work \cite{4-138}, and algebraic equations for buses in $\mathcal{N}_o$, system dynamics can be modelled by the matrix-form semi-explicit DAEs as follows:
\begin{equation}\label{eq-model-1}
    \left\{
    \begin{aligned}
        & E_g \dot{\theta} = \omega \\
        & \dot{\omega} \!=\! - {M}^{-1} D \omega \!+\! {M}^{-1} E_g \tilde{p} -  {M}^{-1} E_g  E_{n} B \sin(E_{n}^T \theta)  \\
        & E_l \dot{\theta} = D_l^{-1} E_l \tilde{p} - D_l^{-1} E_l  E_{n} B \sin(E_{n}^T \theta)\\
        & 0 = E_o \tilde{p} -  E_o  E_{n} B \sin(E_{n}^T \theta)
    \end{aligned}
    \right.
\end{equation}
where $\tilde{p} \!=\! \text{col}( p_i \!-\! V_i^2 \underline{b}_i) \!\in\! \mathbb{R}^{n_a}$ with $i \!\in\! \mathcal{N}$ and $V_i$ assumed to be constant in the study of angle stability and frequency stability.

\textit{Disturbances}. Four typical types of disturbances, including the power-step disturbance, power-ramp disturbance, power-fluctuation disturbance and three-phase short circuit are considered for DID. The disturbance set is denoted by $\mathcal{D} = \mathcal{D}_1 \cup  \mathcal{D}_2 \cup \mathcal{D}_3 \cup \mathcal{D}_4$, with $\mathcal{D}_1$ to $\mathcal{D}_4$ representing the sets of each type of disturbances, respectively. 

%Among each $\mathcal{D}_i$, locations of disturbances are different. Although multiple disturbances are considered, we do not intend to address the robust or stochastic allocation of virtual inertia and damping in this work. For practical DID, consideration of uncertainties from type and location of disturbances is indispensable, while that of other uncertainties is mostly the icing on the cake.

\subsection{DO Formulation of DID}\label{sec-6-DID-formulation}
Considering the disturbances $\mathcal{D}$ and system dynamics over finite-time horizon $t \!\!\in\!\! [t_0^+, t_f]$, the DID problem can be formulated as the dynamic optimization (DO) model as follows:
\begin{subequations}\label{eq-DID-model-1}
	\begin{align}
        (\text{P}1)~{\min_{M, D, \theta^k, \omega^k}} ~&  J  = \sum_{k \in \mathcal{D}} J^k \label{eq-DID-model-1:1}\\ 
	\st_{\forall k \in \mathcal{D}}
    ~& \text{Eq.} (\ref{eq-model-1}) \text{~for~} M, D, \theta^k, \omega^k, B^k, \tilde{p}^k
    \label{eq-DID-model-1:2}\\
	~& \omega^k(t_0) = \omega_{t_0} \label{eq-DID-model-1:3}  \\
    ~& E_{gl} \theta^k(t_0) = E_{gl} \theta_{t_0} \label{eq-DID-model-1:4} \\
    ~& \underline{\omega}^k \leq \omega^k  \leq \overline{\omega}^k   \label{eq-DID-model-1:5} \\
    ~& - \overline{\delta} \leq E_n^T \theta^k  \leq \overline{\delta}   \label{eq-DID-model-1:5-1} \\
    ~& \underline{p}_g \leq E_g p^k - M \dot{\omega}^k - D \omega^k \leq \overline{p}_g  \label{eq-DID-model-1:6}\\
    ~& \underline{m} \leq M\mathbbm{1} \leq \overline{m} \label{eq-DID-model-1:7}\\
    ~& \underline{d} \leq D\mathbbm{1} \leq \overline{d} \label{eq-DID-model-1:8}
	\end{align}
\end{subequations}
with
\begin{equation}
    J^k =  \int_{t_0}^{t_f} \left(L_p(\theta^k, \omega^k)  + L_e(M, D, \theta^k, \omega^k)  \right) \text{d} t 
\end{equation}
\vspace{-5pt}
\begin{equation}
    L_p(\theta^k, \omega^k) = \theta^{k T} W_1^k  \theta^{k} + \omega^{k T} W_2^k  \omega^{k} + \dot{\omega}^{k T} W_3^k \dot{\omega}^{k}
\end{equation}
\begin{equation}
    L_e(M\!,\! D\!,\! \theta^k\!,\! \omega^k) \!\!=\!\! (\!M \dot{\omega}^k \!\!+\!\! D \omega^k \!)^{\!T} \!W_{\!4}^{\!k} \!(\!M\! \dot{\omega}^k \!\!+\!\! D \omega^k ) \!+\! \dot{\theta}^{k T} \!W_{\!5}^{\!k} \dot{\theta}^{k}
\end{equation}
where $J$ is an integral cost function consisting of system performance term $L_p$ and control effort term $L_e$ for each disturbance; $L_p$ and $L_e$ are both in squared $\mathcal{H}_2$-norm forms, in which $L_p$ measures phase cohesiveness, frequency boundedness and frequency oscillation, and $L_e$ measures control efforts of generators and loads; (\ref{eq-DID-model-1:2}) are system DAE constraints, (\ref{eq-DID-model-1:3}) and (\ref{eq-DID-model-1:4}) are initial value constraints for differential variables; (\ref{eq-DID-model-1:5}) are system frequency constraints; (\ref{eq-DID-model-1:5-1}) are transient stability constraints; (\ref{eq-DID-model-1:6}) are generator power constraints; (\ref{eq-DID-model-1:7}) and (\ref{eq-DID-model-1:8}) are bound constraints for inertia and damping coefficients, respectively. In addition, superscript $k$ indicates that parameters or variables correspond to disturbance $k$; $W_1^k = \rho_k E_n^T W_1 E_n$, $W_2^k = \rho_k W_2$, $W_3^k = \rho_k W_3$, $W_4^k = \rho_k W_4$, $W_5^k = \rho_k E_l^T D_l W_5 D_l E_l $, with $\rho_k $ being the weight coefficient of disturbance $k$ and $W_1$ to $W_5$ being diagonal weight matrices; $\underline{\omega}^k$ and $\overline{\omega}^k$ are determined by grid codes, which generally disturbance type-dependent and piecewise-constant \cite{4-243}; $\underline{m}_i = \overline{m}_i$ for $i \in \mathcal{N}_{v_o} \cup \mathcal{N}_{v_d}$ and $\underline{d}_i = \overline{d}_i$ for $i \in \mathcal{N}_{v_o} \cup \mathcal{N}_{v_m}$. Note that $\theta^k$, $\omega^k$, $B^k$, $\tilde{p}^k$, $p^k$, $\underline{\omega}^k$ and $\overline{\omega}^k$ are all time-variant variables or parameters.

\hlr{
\begin{remark}
    Since this work focuses more on solution approach for the DID model, we adopt simple system dynamic models and ignore some practical issues in the DID model. For example, the governor dynamics of synchronous generators are ignored, and limits on the maximal value of the rate of change of frequency and total energy released by inverters in response to a disturbance are not considered. However, the proposed approach in the following is still applicable to the case where these practical issues are considered since they generally cause no new non-convexity and nonlinearity.
\end{remark}
}

\subsection{NLP Formulation of DID} \label{sec-nlp-1}
The DO formulation of DID, being intractable to be solved directly, is further transcribed into a finite-dimensional NLP by collocation methods. Seeing that slow dynamics of synchronous generators and fast dynamics of inverters with small inertia coefficients can constitute stiff DAEs, Radua collocation is employed in our work \cite{4-241}. This method corresponds to the fully implicit Runge-Kutta method and has similar high-order accuracy and stability properties. For stiff power grids, Radua method is L-stable. Furthermore, in Radua collocation, endpoints are collection points, which allows constraints to be set easily at the end of each element \cite{4-241}.

Specifically, for each disturbance $k \in \mathcal{D}$, the time interval $[t_0, t_f]$ is first divided into $n_t^k$ finite elements of length $h_i^k$, denoted as $\mathbb{T}^k = \{1,...,n_t^k\}$, such that $\sum_{i=1}^{n_t^k} h_i^k =t_f - t_0$. Then the solution of DAEs can be approximated by Lagrange polynomials over each element. For any given $k \in \mathcal{D}$ and $i \in \mathbb{T}^k$, we have
\begin{equation}\label{eq-collocation-1}
	\left.
        \begin{aligned}
            & t = t_{i-1} + h_i^k \tau\\
            & \hat{\omega}^k (t) = {\Omega}_i^k \ell(\tau) = {\ell}_{\omega}(\tau) \bm{\omega}_i^k  \\
            & \hat{\theta}^k (t) = {\Theta}_i^k \ell(\tau) = {\ell}_{\theta}(\tau) \bm{\theta}_i^k
        \end{aligned}
    \right\}~~~~ t \in [t_{i-1}, t_i], \tau \in [0,1]
\end{equation}
where $\hat{\omega}^k(t) \in \mathbb{R}^{n_g}$ and $\hat{\theta}^k(t) \in \mathbb{R}^{n_a}$ are vectors of piecewise polynomial with $(n_c + 1)$ degree, approximating $\omega^k(t)$ and $\theta^k(t)$, respectively; $\Omega_i^k \!=\! \text{row} (\omega_{ij}^{k}) \!\in\! \mathbb{R}^{n_g \times (n_c \!+\! 1)}$ and $\Theta_i^k \!=\! \text{row} (\theta_{ij}^{k}) \!\in\! \mathbb{R}^{n_a \times (n_c + 1)}$, with $j \in \{0,..., n_c\}$; $\bm{\omega}_i^k = \col(\omega_{ij}^k) \in \mathbb{R}^{n_g(n_c + 1)}$ and $\bm{\theta}_i^k = \col(\Theta_{ij}^k) \!\in\! \mathbb{R}^{n_a(n_c + 1)}$ with $j \in \{0,...,n_c\}$, are vectorization of collocation coefficient matrices; $\ell(\tau) \!=\! \text{col}(\ell_j(\tau)) \in \mathbb{R}^{n_c + 1} $ with $j \in \{0,...,n_c\}$ and $\ell_j(\tau)$ being the Lagrange polynomial with order $(n_c + 1)$, written as
\begin{equation}\label{eq-collocation-3}
    \ell_j(\tau) =  \prod_{r=0, r\ne j}^{n_c} \frac{\tau - \tau_r}{\tau_j - \tau_r}
\end{equation}
with $\tau_j$ representing location of collocation points within each time element, $\tau_0 \!\!=\! 0$, $\tau_{n_c} \!\!=\!\! 1$, and $\tau_j \!<\! \tau_{j+1}$, $j \!\in\! \{0,...,n_c - 1\}$; ${\ell}_{\omega}(\tau) \!\!=\! \ell(\tau)^T \!\otimes\! I_{n_g} $ and ${\ell}_{\theta}(\tau) \!=\! \ell(\tau)^T \!\otimes\! I_{n_a} $. For the convenience of NLP formulation, we mainly used the second expression of $\hat{\omega}^k$ and $\hat{\theta}^k$. Additionally, we will use the derivative of $\ell_j(\tau)$, and in matrix form, we have $\dot{{\ell}}_{\omega}(\tau) \!=\! \dot{\ell}(\tau)^T \!\otimes\! I_{n_g} $ and $\dot{{\ell}}_{\theta}(\tau) \!=\! \dot{\ell}(\tau)^T \!\otimes\! I_{n_a} $, with $\dot{\ell}(\tau) \!=\!  \frac{\text{d} \ell(\tau)}{\text{d} \tau} $, $ \left[\!\frac{\text{d} \ell(\tau)}{\text{d} \tau}\!\right]_j \!\!=\!\!  \frac{\text{d} \ell_j(\tau)}{\text{d} \tau} $. 

With $\omega^k$ and $\theta^k$ replaced by $\hat{\omega}^k$ and $\hat{\theta}^k$ respectively, $J_i^k$ can be approximated by $\hat{J}_i^k$ given as
\begin{equation}\label{eq-collocation-3-1}
    \!\!\!\!\!\begin{aligned}
        \hat{J}_i^k &\!\!=\! \int_{t_{i-1}}^{t_{i-1} + h_i^k} 
        L_p(\hat{\theta}^k, \hat{\omega}^k) + L_e(M, D, \hat{\theta}^k, \hat{\omega}^k) \text{d} t \\
        = & h_i^k \!\!\! \int_0^1 \!\!\!
        L_p({\ell}_{\theta}(\!\tau\!) \bm{\theta}_i^k\!, {\ell}_{\omega}(\!\tau\!) \bm{\omega}_i^k) \!\!+\!\! L_e(M\!,\! D\!,\! {\ell}_{\theta}(\!\tau\!) \bm{\theta}_i^k\!,\! {\ell}_{\omega}(\!\tau\!) \bm{\omega}_i^k) \text{d} \tau
    \end{aligned}
\end{equation}
Since $L_p$ and $L_e$ are both in the form of polynomial w.r.t time $t$ or $\tau$, definite integral of them can be computed analytically and thus, $\hat{J}_i^k$ can be explicitly formulated as a polynomial function in terms of $M$, $D$, $\bm{\theta}_i^k$ and $\bm{\omega}_i^k$, given by
\begin{equation}\label{eq-collocation-3-2}
    \begin{aligned}
        \hat{J}_i^k\!(\!M\!,\! D\!,\! \bm{\theta}_i^k\!,\! \bm{\omega}_i^k\!)\! =& \bm{\theta}_{i}^{k T} \! \bm{W}_{\!1i}^{\!k} \bm{\theta}_{i}^k 
        \!\!+\!\! \bm{\omega}_{i}^{\!k T} \! \bm{W}_{2i}^k \bm{\omega}_{i}^k 
        \!\!+\!\! \bm{\omega}_{i}^{\!k T} \!\! \check{M} \! \bm{W}_{\!3i}^{\!k} \check{M} \! \bm{\omega}_{i}^k  \\
        &\!+\! \bm{\omega}_{i}^{k T} \check{M}  \bm{W}_{4i}^k \check{D} \bm{\omega}_{i}^k + \bm{\omega}_{i}^{k T} \check{D}  \bm{W}_{5i}^k \check{D} \bm{\omega}_{i}^k 
    \end{aligned}
\end{equation}
where matrices $\bm{W}_{1i}^k$ to $\bm{W}_{5i}^k$ are given in Table \ref{tb-appendix-1}.

Lagrange polynomial $\ell_j$ satisfies $\ell_j(\tau_r) = \delta_{jr}$ for $\forall j,r \in \{0,...,n_c\}$, with $\delta_{jr}$ being the Kronecker delta, and thus $\hat{\omega}^k$ and $\hat{\theta}^k$ have the property that $\hat{\omega}^k(t_{i-1} + \tau_j h_i^k) = \omega_{ij}^k$ and $\hat{\theta}^k(t_{i-1} + \tau_j h_i^k) = \theta_{ij}^k$, respectively. With this property, substituting the polynomial into DAE constraints (\ref{eq-DID-model-1:2}) and enforcing the resulting algebraic equations at the interpolation points $\tau_r$ lead to the collocation equations for DAE constraints as follows:
\begin{equation}\label{eq-collocation-4}
	\!\!\!\!\! \left\{
		\begin{aligned}
            & E_g \dot{\ell}_{\theta}(\tau_r) \bm{\theta}_i^k  = h_i^k \omega_{ir}^k ~~~ \forall r \in \{1,...,n_c\} \\
            & \dot{\ell}_{\omega}(\tau_r) \bm{\omega}_i^k  \!\!=\!\! - h_i^k {M}^{-1} D \omega_{ir}^k + h_i^k {M}^{-1} E_g \tilde{p}_{ir}^k \\
            &~~~~~~~~- h_i^k {M}^{-\!1} \!E_g  E_{n} B^k_{ir} \sin(E_{n}^T \theta_{ir}^k )  ~~ \forall r \!\!\in\!\! \{\!1,...,n_c\!\} \\
            & E_l \dot{\ell}_{\theta} (\tau_r) \bm{\theta}_i^k  \!\!=\!\! h_i^k D_l^{-1} E_l \tilde{p}_{ir}^k \\
            &~~~~~~~~~- h_i^k D_l^{-1} E_l  E_{n} B^k_{ir} \sin(E_{n}^T \theta_{ir}^k) ~~ \forall r \!\!\in\!\! \{\!1,...,n_c\!\} \\
            & 0 = E_o \tilde{p}_{ir}^k -  E_o  E_{n} B_{ir}^k \sin(E_{n}^T \theta_{ir}^k) ~~ \forall r \in \{0,...,n_c\}
		\end{aligned}
	\right.
\end{equation} 

Additionally, $\omega_{i0}^k$ and $E_{gl} \theta_{i0}^k$ are determined by initial conditions (\ref{eq-DID-model-1:3}) and (\ref{eq-DID-model-1:4}), or enforced by the continuity of the differential variable profiles across element boundaries as follows:
\begin{equation}\label{eq-collocation-5} 
    \!\!\!\!\!\!\!\!\! \left[\!\!
    \begin{aligned}
        &\omega_{i0}^k\\
        &E_{\!gl} \theta_{i0}^k
    \end{aligned}
    \!\right]    \!\!\!=\!\! 
	\left\{\!\!\!\!
		\begin{aligned}
			&     \left[
                \begin{aligned}
                    &\omega_{t_0}\\
                    &E_{gl} \theta_{t_0}
                \end{aligned}
                \right]   i=1 \\
            & \left[
                \begin{aligned}
                &\hat{\omega}^k(t_{i-2} \!+\! \tau_{n_c} h_{i-1}^k) \!=\! \omega_{(i-1)n_c}^k \\
                &E_{gl} \hat{\theta}^k\!(\!t_{i-\!2} \!+\!\! \tau_{\!n_c} \! h_{i\!-\!1}^k)  \!\!=\!\! E_{gl} \theta_{(i\!-\!1)n_c}^k
            \end{aligned}
            \right] i \!\!\in\!\! \{\!2,\!...,\!n_t^k\!\}
		\end{aligned}
	\right. \!\!\!\!\!\!\!
\end{equation}

Enforcing path constraints (\ref{eq-DID-model-1:5}) to (\ref{eq-DID-model-1:6}) at all collocation points gives
\begin{equation}\label{eq-collocation-7}
    \underline{\omega}^k_{ir} \leq \omega^k_{ir}  \leq \overline{\omega}^k_{ir} ~~~~ \forall r \in \{1,...,n_c\}
\end{equation}
\begin{equation}\label{eq-collocation-7-1}
    - \overline{\delta}\leq \theta^k_{ir}  \leq \overline{\delta} ~~~~ \forall r \in \{1,...,n_c\}
\end{equation}
\begin{equation}\label{eq-collocation-8}
    \underline{p}_g \!\leq\! E_g p^k_{ir} \!-\! \frac{1}{h_i^k} M \dot{\ell}_{\theta}(\tau_r) \bm{\omega}_i^k \!-\! D \omega^k_{ir} \!\leq\! \overline{p}_g ~~~~ \forall r \!\!\in\! \{1,...,n_c\}
\end{equation} 

Finally, by combining (\ref{eq-DID-model-1:7}), (\ref{eq-DID-model-1:8}) and (\ref{eq-collocation-3-2}) to (\ref{eq-collocation-8}), we can derive the NLP formulation of DID given by
\begin{subequations}\label{eq-collocation-9}
	\begin{align}
		(\text{P}2) & {\min_{M, D, \bm{\omega}, \bm{\theta}}}  \hat{J} (M\!, D\!, \bm{\theta}\!, \bm{\omega})
         \!=\!\!\!\!\!\!\!  \sum_{k\in\mathcal{D}, i\in\mathbb{T}^k} \!\!\! \hat{J}_i^k(M, \!D, \! \bm{\theta}_i^k, \bm{\omega}_i^k)  \label{eq-collocation-9:1}\\ 
	~\st & _{\forall k \in \mathcal{D}, i \in \mathbb{T}^k} \nonumber\\
    ~&  \bm{C}_{1i}^k \sin ( \bm{A}_{1i}^k \bm{\theta}_i^k)  \!\!+\!\!  \check{M} \bm{H}_{1i}^k \bm{\omega}_i^k \!\!+\!\!  \check{D} \bm{H}_{2i}^k \bm{\omega}_i^k \!\!+\!\! \bm{b}_{1i}^k  \!=\! \bm{0}   \label{eq-collocation-9:2} \\
    ~& \bm{C}_{2i}^k \sin ( \bm{A}_{1i}^k \bm{\theta}_i^k)  \!+\! \bm{A}_{2i}^k \bm{\theta}_i^k \!+\! \bm{A}_{3i}^k \bm{\omega}_i^k + \bm{b}_{2i}^k  = \bm{0}   \label{eq-collocation-9:2-1} \\
    ~&  \underline{\bm{c}}_{1i}^k 
     \leq \check{M} \bm{H}_{1i}^k \bm{\omega}_i^k \!+\!  \check{D} \bm{H}_{2i}^k \bm{\omega}_i^k  \leq \overline{\bm{c}}_{1i}^k\label{eq-collocation-9:5} \\
     ~&  \underline{\bm{c}}_{2i}^k 
     \!\leq\! \bm{A}_{4i}^k \bm{\omega}_i^k \!+\! \bm{A}_{5i}^k \bm{\theta}_i^k \!+\!  \bm{A}_{6i}^k M \mathbbm{1} \!+\! \bm{A}_{7i}^k D \mathbbm{1} \! \leq \! \overline{\bm{c}}_{2i}^k\label{eq-collocation-9:5-1} \\
     ~&  \bm{L}_1 \bm{\theta}_i^k - \bm{L}_2 \bm{\theta}_{i-1}^k = \bm{0} \label{eq-collocation-9:3} \\
     ~&  \bm{L}_3 \bm{\omega}_i^k - \bm{L}_4 \bm{\omega}_{i-1}^k = \bm{0} \label{eq-collocation-9:4} 
 	\end{align}
\end{subequations}
where matrices $\bm{C}_{1i}^k$,  $\bm{C}_{2i}^k$, $\bm{A}_{1i}^k$ to $\bm{A}_{7i}^k$, $\bm{H}_{1i}^k$, $\bm{H}_{2i}^k$, $\bm{L}_1$ to $\bm{L}_4$, $\bm{b}_{1i}^k$, $\bm{b}_{2i}^k$,$\underline{\bm{c}}_{1i}^k$, $\overline{\bm{c}}_{2i}^k$, $\underline{\bm{c}}_{2i}^k$  and  $\overline{\bm{c}}_{2i}^k$ are given in Table \ref{tb-appendix-1}.

\begin{table}[h]\centering
	\caption{\hlr{Matrices in the NLP formulation of DID}}
    \label{tb-appendix-1}
    {\footnotesize{
    \begin{tabular*}{\hsize}{@{\extracolsep{\fill}} cccccc}
        \toprule
        $\bm{A}_{2i}^k$  & $\bm{A}_{3i}^k$ & $\bm{A}_{4i}^k$  & $\bm{A}_{5i}^k$  & $\bm{A}_{6i}^k$ & $\bm{A}_{7i}^k$ \\
        \midrule
        $\left[\!\!\!\! \begin{array}{ccc}
            \col(E_g \dot{\ell}_\theta (\tau_{r_1})) \\
             \col(D_l E_l  \dot{\ell}_\theta(\tau_{r_1})) \\
            O \\
            \end{array} 
        \!\!\!\!\right]\!\!\!\!\!\!$
        &
        $\left[\!\!\! \begin{array}{ccc}
            - E_2   h_i^k \\
             O \\
            O \\
            \end{array} 
        \!\!\!\right]\!\!\!\!\!\!$
        &
        $\left[\!\! \begin{array}{ccc}
            E_2\\
            O\\
            O\\
            O
            \end{array} 
        \!\!\right]\!\!\!\!\!\!$
        &
        $\left[\!\! \begin{array}{ccc}
            O\\
            E_3\\
            O\\
            O
            \end{array} 
        \!\!\right]\!\!\!\!\!\!$
        &
        $\left[\!\! \begin{array}{ccc}
            O\\
            O\\
            I_{n_g} \\
            O
            \end{array} 
        \!\!\right]\!\!\!\!\!\!$
        &
        $\left[\!\! \begin{array}{ccc}
            O\\
            O\\
            O\\
            I_{n_g}
            \end{array} 
        \!\!\right]$
\end{tabular*}
\begin{tabular*}{\hsize}{@{\extracolsep{\fill}} cccc}
    \toprule
    $\bm{A}_{1i}^k$ & $\bm{H}_{1i}^k$ & $\bm{H}_{2i}^k$ & $\bm{C}_{1i}^k$   \\
    \midrule
    $ \diag(E_{n}^T)$ 
    &
    $\col(\dot{\ell}_\omega(\tau_{r_1}))$
    &
    $ h_i^k E_2 $
    &
    $ h_i^k \diag( E_g  E_{n} B^k_{ir_1})  E_1 $   
\end{tabular*}
\begin{tabular*}{\hsize}{@{\extracolsep{\fill}} cccc}
    \toprule
    $\bm{C}_{2i}^k$ &  $\bm{b}_{2i}^k$ & $\underline{\bm{c}}_{2i}^k$ &  $\overline{\bm{c}}_{2i}^k$ \\
    \midrule
    $\left[\!\!\!\! \begin{array}{ccc}
        O\\
        h_i^k \diag( E_l  E_{n} B^k_{ir_1}) E_1\\
        -  \diag( E_o  E_{n} B^k_{ir_0})   E_1 \\
        \end{array} 
    \!\!\!\!\right] \!\!\!\!\!\!\!\!\!$    
    &
    $\!\left[\!\!\!\! \begin{array}{c}
        \bm{0} \\
        \!-\! \col(h_i^k E_l \tilde{p}_{ir_1}^k) \\
        \col(E_o \tilde{p}_{ir_0}^k) \\
        \end{array} 
    \!\!\!\!\!\right] \!\!\!\!\!\!\!\!$
    &
    $\!\!\left[ \!\!\!\!\begin{array}{c}
        \col(\underline{\omega}_{i r_1}^k)\\
        \col(-\overline{\delta} )\\
        \underline{m}\\
        \underline{d}
        \end{array} \!\!
    \!\!\!\right]\!\!\!\!\!\!\!\!$
    &
    $\!\!\left[ \!\!\!\!\begin{array}{c}
        \col(\overline{\omega}_{i r_1}^k)\\
        \col(\overline{\delta} )\\
        \overline{m}\\
        \overline{d}
        \end{array} \!\!\!\!
    \right]$
\end{tabular*}
\begin{tabular*}{\hsize}{@{\extracolsep{\fill}} cccc}
    \toprule
    $\bm{b}_{1i}^k$ & $\underline{\bm{c}}_{1i}^k$ &  $\overline{\bm{c}}_{1i}^k$  \\
    \midrule
    $\!\!\!\! - \col( h_i^k E_g \tilde{p}_{ir_1}^k)\!\!\!\! $ 
    &
    $ h_i^k \col(E_g p_{i r_1}^k \!\!\!-\!\! \overline{p}_g)  $
    &
    $\!\!\! h_i^k \col(E_g p_{i r_1}^k \!\!-\!\! \underline{p}_g ) \!\!\!\!\!\!$
\end{tabular*}
\begin{tabular*}{\hsize}{@{\extracolsep{\fill}} ccc}
    \toprule
    $\bm{L}_1$  &  $\bm{L}_2$ &   $\bm{L}_3$ \\
    \midrule
    $E_{gl}\! \row(\!I_{n_a}\!,\! O_{n_a \!\times\! n_a n_c}\!)\!\!\!\!\!\!$ 
    &
    $E_{gl} \!\row\!(\!O_{n_a \!\times\! n_a n_c}\!,\! I_{n_a}\!)\!\!\!\!\!\!$
    &
    $\row\!(\!I_{n_g}\!,\! O_{n_g \!\times\! n_g n_c}\!)\!\!$
\end{tabular*}
\begin{tabular*}{\hsize}{@{\extracolsep{\fill}} ccc}
    \toprule
    $\bm{L}_4$ & $E_1$  &  $E_2$\\
    \midrule
    $\row\!(\!O_{n_g \times n_g n_c}\!,\! I_{n_g}\!)\!\!\!\!$
    &
    $\row\!(\!O_{n_b n_c \times n_b}\!,\! I_{n_b n_c}\!)\!\!\!\!$ 
    &
    $\row\!(\!O_{n_g n_c \!\times n_g }\!,\! I_{n_g n_c} \!)\!\!$
\end{tabular*}   
\begin{tabular*}{\hsize}{@{\extracolsep{\fill}} ccc}
    \toprule
    $E_3$ & $r_0$  &  $r_1$\\
    \midrule
    $\row(O_{n_b n_c \times n_a }, \diag(E_n^T))$
    &
    $\{0,...,n_c\}$
    &
    $\{1,...,n_c\}$
\end{tabular*}  
\begin{tabular*}{\hsize}{@{\extracolsep{\fill}} cc}
    \toprule
    $\bm{W}_{1i}^k$  &  $\bm{W}_{2i}^k$ \\
    \midrule
    $h_i^k (S_1 \otimes W_2^k) + \frac{1}{h_i^k} (S_2 \otimes W_3^k)$
    &
    $\frac{1}{h_i^k} (S_2 \otimes W_4^k)$ 
\end{tabular*} 
\begin{tabular*}{\hsize}{@{\extracolsep{\fill}} ccc}
    \toprule
    $\bm{W}_{3i}^k$ &  $\bm{W}_{4i}^k$ &  $\bm{W}_{5i}^k$  \\
    \midrule
    $h_i^k (S_1 \otimes W_1^k) +  \frac{1}{h_i^k} (S_2 \otimes W_5^k)$ 
    &
    $\!\!2 (S_3 \otimes W_4^k)$
    &
    $\!\!\!h_i^k (S_1 \otimes W_4^k)$
\end{tabular*} 
\begin{tabular*}{\hsize}{@{\extracolsep{\fill}} ccc}
    \toprule
    $S_1$ &  $S_2$ & $S_3$  \\
    \midrule
    $\int_{0}^{1} \ell(\tau) \ell(\tau)^T  \text{d}\tau$
    &
    $\int_{0}^{1} \dot{\ell}(\tau) \dot{\ell}(\tau)^T \text{d}\tau$
    &
    $\int_{0}^{1} \dot{\ell}(\tau) \ell(\tau)^T  \text{d}\tau$
    \\
    \bottomrule
\end{tabular*} 
    }}
\end{table}

\section{SDP Relaxation for the NLP of DID}\label{section-2}

The NLP of DID is non-convex due to not only the quadratic interval constraints and equality constraints that contain sine and quadratic terms, but also the non-convex polynomial objective function. Solving the NLP directly by local optimizaition approach will very likey produce inferior solutions in terms of global optimality. Therefore, SDP-based convex relaxation of the NLP of DID is derived in this section, aiming to enhance global optimality of final solutions.

\subsection{SDP Relaxation of DID with Linearized Power Flow}

We first consider a simpler case where power flow terms in (\ref{eq-model-1}) are linearized. This case is valid when $\mathcal{D}$ contains no large disturbances. For the cost function, we define the lifting as follows:
\begin{equation}\label{eq-lifting-1}
    [\bm{l}_{\text{m} i}^k, \bm{l}_{\text{d} i}^k ] =  [\check{M} \bm{\omega}_i^k, \check{D} \bm{\omega}_i^k] ~~~ \forall k \in \mathcal{D}, i \in \mathbb{T}^k.
\end{equation}
which replaces products with novel linear variables. Then problem (P2) with linearization and lifting (\ref{eq-lifting-1}) can be equivalently formulated as the following QCQP problem:
\begin{subequations}\label{eq-qcqp-1}
	\begin{align}
        (\text{P}3)~{\min_{\bm{x}}} ~~~~~ & 
        \sum_{k \in \mathcal{D}, i \in \mathbb{T}^k} [\bm{x}]_i^{kT} \bm{P}_{0i}^k [\bm{x}]_i^k \label{eq-qcqp-1:1}\\
        \mathop{\mathrm{s.t.}}\limits_{\forall k \in \mathcal{D}, i \in \mathbb{T}^k}
    ~& \text{col}( [\bm{x}]_i^{kT} \bm{P}_{(\hbar, r,j) i}^k [\bm{x}]_i^k ) \!\!+\!\! \bm{Q}_i^{k}  [\bm{x}]_i^k \!\!+\!\! \tilde{\bm{b}}_{1i}^k  \!=\! \bm{0}  \label{eq-qcqp-1:2} \\
    ~& \bm{A}_i^k [\bm{x}]_i^k \leq \bm{b}_i^k \label{eq-qcqp-1:4}\\
    ~&  \bm{L}_{13} [\bm{x}]_i^k - \bm{L}_{14} [\bm{x}]_{i-1}^k    = \bm{0} \label{eq-qcqp-1:5} 
 	\end{align}
\end{subequations}
Here for $(\hbar, r,j)$, $\hbar \!\in\! \{1, 2, 3\}$, $r \in \{1,...,n_c\}$ for $\hbar=1$ and $\{0,...,n_c\}$ for $\hbar=2 ~\text{or} ~3$,  and $j \in \mathcal{N}_g$; (\ref{eq-qcqp-1:2}) corresponds to lifting equalities (\ref{eq-lifting-1}) and (\ref{eq-collocation-9:2}) in (P2), (\ref{eq-qcqp-1:4}) corresponds to (\ref {eq-collocation-9:2-1}) to (\ref{eq-collocation-9:5-1}) in (P2), and (\ref{eq-qcqp-1:5}) corresponds to (\ref{eq-collocation-9:3}) and (\ref{eq-collocation-9:4}) in (P2). Note that (\ref{eq-collocation-9:5-1}) in (P2) is in quadratic form but transformed into linear form in (P3). Coefficient matrices in (P3) are given in Table \ref{tb-appendix-2}. By further introducing a matrix variable $\bm{X} \!=\! \bm{x} \bm{x}^T$, (P3) is reformulated as
\begin{subequations}\label{eq-sdp-1}
	\begin{align}
        (\text{P}4)~{\min_{ \bm{Z} }} ~~~~& 
        \sum_{k \in \mathcal{D}, i \in \mathbb{T}^k} \Tr(\bm{P}_{0i}^k [\bm{X}]_i^k) \label{eq-sdp-1:1}\\
    \st_{\forall k \in \mathcal{D}, i \in \mathbb{T}^k}
    ~& \text{col}( \Tr(\bm{P}_{(\hbar, r,j) i}^k  [\bm{X}]_i^k) ) \!+\! \bm{Q}_i^{k}  [\bm{x}]_i^k \!+\! \tilde{\bm{b}}_{1i}^k  \!\!=\!\! \bm{0}  \label{eq-sdp-1:2} \\
    ~& \bm{A}_i^k [\bm{x}]_i^k \leq \bm{b}_i^k, \bm{L}_{13} [\bm{x}]_i^k - \bm{L}_{14} [\bm{x}]_{i - 1}^k    = \bm{0} \label{eq-sdp-1:4}\\
    % ~&   \label{eq-sdp-1:5} \\
    ~& \bm{Z} = \left[\!\! \begin{array}{cc}
        \bm{X}   \!\!\!&\!\! \bm{x} \\
        \bm{x}^T \!\!\!&\!\! 1
    \end{array} \!\!\right] \succeq 0 \label{eq-sdp-1:6}\\
    ~&  \rank(\bm{Z}) = 1 \label{eq-sdp-1:7}
 	\end{align}
\end{subequations}
Dropping the non-convex rank constraint (\ref{eq-sdp-1:7}) gives the convex SDP as follows:
\begin{equation}\label{eq-sdp-3}
	(\text{P}5)~ \text{(\ref{eq-sdp-1:1})} \sim \text{(\ref{eq-sdp-1:6})}
\end{equation}
We call (P5) a SDP relaxation of (P3) since the feasible region of (P3) or (P4) is a subset of the feasible region of (P5). If an optimal point $\bm{Z}^*$ of (P5) satisfies $\text{rank}(\bm{Z}^*) = 1$, i.e., $\bm{X}^* \!=\!  \bm{x}^* \bm{x}^{*T}$, then $\bm{x}^*$ is also the global optimum of (P3). 

\begin{table}[h]
	\caption{\hlr{Matrices in the QCQP formulation of DID}}
    \label{tb-appendix-2}
    {\footnotesize{
\begin{tabular*}{\hsize}{@{\extracolsep{\fill}} cc}
	\toprule
	 $\bm{P}_{0i}^k$ & $\bm{b}_i^k$ \\
	\midrule
    $\left[ \begin{array}{cccccc}
        \!\! \bm{O} &\!\!\! \bm{O} &\!\! \bm{O}       \!\!\!\!& \bm{O}  &\!\!\!\!\!\! \bm{O}        \!\!\!\!&\!\!\!\! \bm{O} \\ \!\! 
        \!\! \bm{O} \!\!&\!\!\! \bm{O} &\!\! \bm{O}       \!\!\!\!& \bm{O}  &\!\!\!\!\!\! \bm{O}        \!\!\!\!&\!\!\!\! \bm{O} \\ \!\! 
        \!\! \bm{O} \!\!&\!\!\! \bm{O} &\!\! \bm{W}_{1i}^k\!\!\!\!& \bm{O} &\!\!\!\!\!\! \bm{O}  \!\!\!\!&\!\!\!\! \bm{O} \\ \!\! 
        \!\! \bm{O} \!\!&\!\!\! \bm{O} &\!\! \bm{O}       \!\!\!\!& \bm{W}_{2i}^k &\!\!\!\!\!\! \bm{O}  \!\!\!\!&\!\!\!\! \bm{O}  \\ \!\! 
        \!\! \bm{O} \!\!&\!\!\! \bm{O} &\!\! \bm{O}       \!\!\!\!& \bm{O} &\!\!\!\!\!\! \bm{W}_{3i}^k   \!\!\!\!&\!\!\!\!  \frac{1}{2} \bm{W}_{4i}^k \\ \!\! 
        \!\! \bm{O} \!\!&\!\!\! \bm{O} & \bm{O} & \bm{O} &\!\!\!\!  \frac{1}{2} \bm{W}_{4i}^k & \bm{W}_{5i}^k 
        \end{array} 
    \!\!\right]$\!\!\!\!
    &\!\!\!\!\!\!\!\!\!\!\!
    $\!\!\left[ \!\!\begin{array}{c}
        \!\! - (\bm{b}_{2i}^k + \bm{C}_{2i}^k \bm{\Lambda}_{1i}^k) \\
        \!\!   \bm{b}_{2i}^k + \bm{C}_{2i}^k \bm{\Lambda}_{1i}^k \\
        \!\! \overline{\bm{c}}_{1i}^k + \bm{b}_{1i}^k + \bm{C}_{1i}^k \bm{\Lambda}_{1i}^k  \\
        \!\! - (\underline{\bm{c}}_{1i}^k + \bm{b}_{1i}^k + \bm{C}_{1i}^k \bm{\Lambda}_{1i}^k) \\
        \!\! \overline{\bm{c}}_{2i}^k \\
        \!\! - \underline{\bm{c}}_{2i}^k
        \end{array} \!\!\!\!
    \right]$\!\!\!\!\!
    \end{tabular*}
    \begin{tabular*}{\hsize}{@{\extracolsep{\fill}} cc}
        \toprule
         $\bm{Q}_i^k$ & $\tilde{\bm{b}}_{1i}^k$\\
        \midrule  \!\!
        $\left[ \!\!\begin{array}{cccccc}
            \!\! \bm{O} \!\!\!&\!\! \bm{O} \!\!&\!\! \bm{C}_{1i}^k \bm{\Lambda}_{0i}^k \!&\!\!\!\! \bm{O} \!\!\!\!&\!\!\!\! \bm{O} \!\!\!\!&\!\!\!\! \bm{O}  \\\!\!
            \!\! \bm{O} \!\!\!\!&\!\!\!\! \bm{O} \!\!\!\!&\!\! \bm{O}  \!\!\!\!&\!\!\!\! \bm{O} \!&\!\!\!\! -I_{n_g(n_c + 1)}  \!\!\!\!&\!\!\!\! \bm{O} \\ \!\! 
            \!\! \bm{O} \!\!\!\!&\!\!\!\! \bm{O} \!\!\!\!&\!\! \bm{O}  \!\!\!\!&\!\!\!\! \bm{O} \!\!\!\!&\!\!\!\! \bm{O} \!\!\!\!&\!\!\!\! -I_{n_g(n_c + 1)}
            \end{array} \!\!\!\!\!\!
        \right]$\!\!\!\!\!\!
        &\!\!
        $\left[ \!\! \begin{array}{c}
            \bm{b}_{1i}^k \!+\! \bm{C}_{1i}^k \bm{\Lambda}_{1i}^k\\
            \bm{O} \\
            \bm{O}
            \end{array} \!\!
        \right]$
    \end{tabular*}
    \begin{tabular*}{\hsize}{@{\extracolsep{\fill}} ccc}
        \toprule
         $\bm{A}_{i}^k$  & $\bm{L}_{13}^T$ & $\bm{L}_{14}^T$  \\
        \midrule  \!\!
        $\left[ \begin{array}{cccccc}
            \!\! \bm{O} \!\!&\!\! \bm{O} \!\!&\!\!\!\!\!\! \bm{C}_{2i}^k \bm{\Lambda}_{0i}^k + \bm{A}_{2i}^k   &\!\!\!\! \bm{A}_{3i}^k  &\!\!\!\! \bm{O} \!\!\!\!& \bm{O} \\ \!\! 
            \!\! \bm{O} \!\!&\!\! \bm{O} \!\!&\!\!\!\!\!\! -(\bm{C}_{2i}^k \bm{\Lambda}_{0i}^k + \bm{A}_{2i}^k) &\!\!\!\!  - \bm{A}_{3i}^k &\!\!\!\! \bm{O} \!\!\!\!& \bm{O} \\ \!\! 
            \!\! \bm{O} \!\!&\!\! \bm{O} \!\!&\!\!\!\!\!\! -\bm{C}_{1i}^k \bm{\Lambda}_{0i}^k &\!\!\!\! \bm{O} &\!\!\!\! \bm{O}  \!\!\!\!& \bm{O} \\ \!\! 
            \!\! \bm{O} \!\!&\!\! \bm{O} \!\!&\!\!\!\!\!\! \bm{C}_{1i}^k \bm{\Lambda}_{0i}^k  &\!\!\!\! \bm{O} &\!\!\!\! \bm{O}  \!\!\!\!& \bm{O}  \\ \!\! 
            \!\! \bm{A}_{6i}^k \!\!&\!\! \bm{A}_{7i}^k \!\!&\!\!\!\!\!\! \bm{A}_{5i}^k &\!\!\!\! \bm{A}_{4i}^k &\!\!\!\! \bm{O}  \!\!\!\!& \bm{O} \\ \!\! 
            \!\! -\bm{A}_{6i}^k \!\!&\!\! -\bm{A}_{7i}^k \!\!&\!\!\!\!\!\! -\bm{A}_{5i}^k &\!\!\!\! -\bm{A}_{4i}^k &\!\!\!\! \bm{O}  \!\!\!\!& \bm{O} 
            \end{array} 
        \!\!\!\!\right]$\!\!\!\!\!\!\!
        &\!\!\!\!
        $\left[ \!\!\!\begin{array}{cc}
            \!\! \bm{O}  &\!\!\!\! \bm{O} \\
            \!\! \bm{O}  &\!\!\!\! \bm{O} \\
            \!\! \bm{L}_1  &\!\!\!\! \bm{O} \\
            \!\! \bm{O}  &\!\!\!\! \bm{L}_3 \\
            \!\! \bm{O}  &\!\!\!\! \bm{O} \\
            \!\! \bm{O}  &\!\!\!\! \bm{O}
            \end{array} \!\!
        \!\!\!\right]$\!\!\!\!
        &\!\!\!\!\!\!
        $\left[\!\!\!\begin{array}{cc}
            \!\! \bm{O}  &\!\!\!\! \bm{O} \\
            \!\! \bm{O}  & \!\!\!\!\bm{O} \\
            \!\! \bm{L}_2  &\!\!\!\! \bm{O} \\
            \!\! \bm{O}  &\!\!\!\! \bm{L}_4 \\
            \!\! \bm{O}  &\!\!\!\! \bm{O} \\
            \!\! \bm{O}  &\!\!\!\! \bm{O}
            \end{array} \!\!
        \!\!\!\right]$\!\!
    \end{tabular*}
    \begin{tabular*}{\hsize}{@{\extracolsep{\fill}} c}
        \toprule
         $\bm{P}_{(1,r,j) i}^k = ( M \mathbbm{1}, \bm{\omega}_i^k,  \frac{1}{2}\bm{\ell}_{\omega j}( \tau_r) )$, $( D \mathbbm{1}, \bm{\omega}_i^k,   \frac{1}{2}E_{(r,j)} )$\\
         \midrule
         $\bm{P}_{(2,r,j) i}^k = ( M \mathbbm{1}, \bm{\omega}_i^k,   \frac{1}{2} \bm{O}_{(r,j)}^1 )$,
         $\bm{P}_{(3,r,j) i}^k = ( D \mathbbm{1}, \bm{\omega}_i^k,  \frac{1}{2} \bm{O}_{(r,j)}^1 )$ \\
        \bottomrule
    \end{tabular*}
    }}
    \small{Note: Since $\bm{P}_{(\hbar, r, j) i }^k$ are highly sparse and symmetric matrices, they are given in block 3-tuple form and only with lower triangular portions of matrix.}
\end{table}

\subsection{SDP Relaxation of DID with Nonlinear Power Flow}
When disturbance set $\mathcal{D}$ contains large disturbances, linearization of power flow equations can lead to unacceptable approximation errors. Hence fidelity of nonlinearity for disturbances in $\mathcal{D}_4$ should be reserved. We first propose the following quadratic approximation of sine function in domain $[-\theta_c,  \theta_c]$ with $\theta_c \!\in\! [\!\frac{\pi}{2}, \pi]$. Numerical analysis for this approximation can be found in Appendix\ref{appendix_3}.

\begin{approximation}\label{approx-2}
    $\sin \theta \mapsto {\beta}^T {\varsigma}$
    with
    \begin{equation}\label{eq-sin-2}\nonumber
        \begin{aligned}
        ~&\left[\!\!\! \begin{array}{c}
             \vartheta \theta^2 
            + \pi \vartheta \theta 
            + \frac{\pi^2}{4} \vartheta  - 1  \\
            \frac{\sin \theta_b}{\theta_b} \theta  \\
            - \vartheta \theta^2 
            + \pi \vartheta \theta 
            - \frac{\pi^2}{4} \vartheta + 1 \\
            \end{array} \!\!\! \right] \!-\!  {\varsigma} \!=\! {0}~ \text{with}~ \vartheta \!=\! \frac{1 - \sin\theta_b}{(\theta_b - \frac{\pi}{2})^2}, \\
        ~& \theta \!-\! \theta_{cb}^T {\alpha}  \!=\! 0,\!  \mathbbm{1}^T {\beta} \!=\! 1,\! \mathbbm{1}^T {\alpha} \!=\! 1,\! ({\beta}-\mathbbm{1})^2 \!=\! 0,\! {\alpha} \!-\! E {\beta} \!\leq\! 0,\!  {\alpha} \!\geq\! 0  \\
        \end{aligned}
    \end{equation}
    where $\theta \!\!\in\!\! [-\theta_c,  \theta_c]$ with $\theta_c \!\!\in\!\! [\frac{\pi}{2}, \pi]$, ${\beta} \!\!\in\!\! \mathbb{R}^3$, ${\varsigma} \!\!\in\!\! \mathbb{R}^3$, ${\alpha} \!\!\in\!\! \mathbb{R}^4$; $\theta_b \!\!\in\!\! [0, \frac{\pi}{2}]$ and $\theta_c$ are known parameters; $\theta_{cb} \!=\! [-\theta_c, -\theta_b, \theta_b, \theta_c]^T $; and $E = [[1,0,0],[1,1,0],[0,1,1],[0,0,1]]$.
\end{approximation}

Then with the assumption that $\overline{\delta} \in [\frac{\pi}{2}, \pi]$, (P2) with lifting (\ref{eq-lifting-1}), linearization for $k \in \mathcal{D}_1 \cup \mathcal{D}_2 \cup \mathcal{D}_3$ and Approximation \ref{approx-2} for $k \in \mathcal{D}_4$ can also be formulated as a QCQP that we call (P6) hereafter. Variables corresponding to $\alpha$, $\beta$ and $\varsigma$ in Approximation \ref{approx-2}, i.e., $\bm{\alpha}$, $\bm{\beta}$ and $\bm{\varsigma}$, are introduced for (P6). Clearly, $\bm{\alpha}$ only appears in linear form, and thus by introducing $\bm{Y} = \col(\bm{\beta}, \bm{\varsigma}) \cdot \col(\bm{\beta}, \bm{\varsigma})^T$, we obtain the formulation equivalent to (P6) as follows:

%Approximation-related variables, corresponding to $\alpha$, $\beta$ and $\varsigma$ in (\ref{eq-sin-2}), should be introduced in (P6), denoted as $\bm{\alpha}_i^k = \col(\bm{\alpha}_{(r,\imath)i}^k)$, $\bm{\beta}_i^k = \col(\bm{\beta}_{(r,\imath)i}^k)$, $\bm{\varsigma}_i^k = \col(\bm{\varsigma}_{(r,\imath)i}^k)$, with $r \in \{0,...,n_c\}$, $\imath \in \mathcal{N}_b$, $\bm{\alpha}_{(r,\imath)i}^k \in \mathbb{R}^4$, $\bm{\beta}_{(r,\imath)i}^k \in \mathbb{R}^3$, $\bm{\varsigma}_{(r,\imath)i}^k \in \mathbb{R}^3$. Clearly there is no quadratic terms in (P6) coupling $\{\bm{\alpha}_i^k, \bm{\beta}_i^k, \bm{\varsigma}_i^k  \}$ and $[\bm{x}]_i^k$. Additionally, $\bm{\alpha}_i^k$ only appears in linear forms, thus by introducing $\bm{Y}_{(r, \imath)i}^k = \col(\bm{\beta}_{(r, \imath)i}^k, \bm{\varsigma}_{(r, \imath)i}^k) \cdot  \col(\bm{\beta}_{(r, \imath)i}^k, \bm{\varsigma}_{(r, \imath)i}^k)^T$, we obtain the following formulation equivalent to (P6):

\begin{subequations}\label{eq-sdp-4p}
	\begin{align}
        (\text{P}7)~{\min_{\bm{Z}, \tilde{\bm{Z}}}} ~~~~~ & 
        \sum_{k \in \mathcal{D}, i \in \mathbb{T}^k} \Tr(\bm{P}_{0i}^k [\bm{X}]_i^k) \label{eq-sdp-4p:1}\\
    \mathrm{s.t.}
    ~&  \text{(\ref{eq-sdp-1:2})}, \text{(\ref{eq-sdp-1:4})} ~~{\forall k \!\in\! \mathcal{D}_1 \!\cup\! \mathcal{D}_2 \!\cup\! \mathcal{D}_3 , i \!\in\! \mathbb{T}^k} \label{eq-sdp-4p:2}\\
    ~& \text{(\ref{eq-sdp-4:2})} \sim \text{(\ref{eq-sdp-4:5-4})} ~~~{\forall k \!\in\! \mathcal{D}_4, i \!\in\! \mathbb{T}^k} \label{eq-sdp-4p:3}\\
    ~& \tilde{\bm{Z}} \! = \! \left[ \!\!
        \begin{array}{c@{}c@{}}
            \!\!\bm{Y}   \!\!&\!\!\!\!  \begin{array}{c}
                        \bm{\beta} \\
                        \bm{\varsigma} \\
                       \end{array} \\
            \!\! \begin{array}{cc}
                \bm{\beta}^T \!\!&\!\! \bm{\varsigma}^T
            \end{array}  \!\!&\!\!\!\! 1\\
        \end{array}
        \!\! \right] \!\! \succeq 0, \bm{Z} \succeq 0 \label{eq-sdp-4p:4}\\ 
    ~&  \rank( \tilde{\bm{Z}} ) = 1, \rank( \bm{Z} ) = 1 \label{eq-sdp-4p:5}
 	\end{align}
\end{subequations}
with (\ref{eq-sdp-4:2}) $\sim$ (\ref{eq-sdp-4:5-4}) written as follows:
\begin{subequations}\label{eq-sdp-4}
	\begin{align}
    ~& \text{col}(\! \Tr(\bm{P}_{(\!1, \tilde{r},j\!) i}^k  \![\!\bm{X}\!]_i^k)\! ) \!+\! \bm{C}_{1i}^k \! \col(\Tr(\! \tilde{\bm{P}}_1\! [\bm{Y}]_{(\!r,\imath\!)i}^k )\!)  \!\!+\!  \bm{b}_{1i}^k  \!=\! \bm{0}  \label{eq-sdp-4:2} \\
    ~& \text{col}( \Tr(\bm{P}_{(\hbar, r,j) i}^k  [\bm{X}]_i^k) ) + \tilde{\bm{Q}}_i^{k} [\bm{x}]_i^k = \bm{0}  \label{eq-sdp-4:3} \\
    ~&  \underline{\bm{c}}_{1i}^k \leq   \text{col}( \Tr(\bm{P}_{(1, \tilde{r},j) i}^k  [\bm{X}]_i^k) )  \leq \overline{\bm{c}}_{1i}^k \label{eq-sdp-4:3-1} \\
    ~& \tilde{\bm{A}}_i^k [\bm{x}]_i^k \leq \tilde{\bm{b}}_i^k, \bm{L}_{13} [\bm{x}]_i^k - \bm{L}_{14} [\bm{x}]_{i - 1}^k = \bm{0} \label{eq-sdp-4:4}\\
    % ~&  \label{eq-sdp-4:5} \\
    ~& \bm{A}_{8i}^k \!\col(\! \Tr(\!\bm{P}_{(\!\tilde{r}, \imath\!)i}^k [\!\bm{X}\!]_i^k ) \!\!+\!\!  \bm{A}_{9i}^k \bm{\theta}_i^k \!\!-\!\! [\bm{\varsigma}]_i^k \!\!+\!\! \bm{A}_{8i}^k \mathbbm{1} \!(\!\frac{\pi^2\!}{4} \vartheta  \!\!-\!\! 1) \!\!=\!\! \bm{0} \label{eq-sdp-4:5-1}\\
    ~& \bm{A}_{1i}^k \bm{\theta}_i^k - \bm{A}_{10i}^k [\bm{\alpha}]_i^k = \bm{0}, \text{col}( \Tr( \tilde{\bm{P}}_2 [\bm{Y}]_{(r,\imath)i}^k  )) \!\!=\!\! 1  \label{eq-sdp-4:5-2}\\
    % ~&  \label{eq-sdp-4:5-21}\\
    ~& \text{col}(\mathbbm{1}^T [\bm{\alpha}]_{(r,\imath)i}^k) \!\!=\!\! 1, \text{col}(\mathbbm{1}^T [\bm{\beta}]_{(r,\imath)i}^k) \!\!=\!\! 1 \label{eq-sdp-4:5-22}\\
    % ~&   \label{eq-sdp-4:5-23}\\
    ~& \text{col}([\bm{\alpha}]_{(r,\imath)i}^k - E [\bm{\beta}]_{(r,\imath)i}^k) \leq 0, \text{col}([\bm{\alpha}]_{(r,\imath)i}^k) \geq 0 \label{eq-sdp-4:5-4}
    % ~&   \label{eq-sdp-4:5-4}
 	\end{align}
\end{subequations}
where $\hbar \in {2,3}$, $\tilde{r} \in \{1,...,n_c\}$, $r \in \{0,...,n_c\}$, $\imath \in \mathcal{B}$ and $j \in \mathcal{N}_g$. Analogously to (P4), dropping rank constraints in (P7) gives the following SDP:
\begin{equation}\label{eq-sdp-5}
	(\text{P}8)~ \text{(\ref{eq-sdp-4p:1}) - (\ref{eq-sdp-4p:4})}
\end{equation}
Since (P5) can be regarded as a special case of (P8), only (P8) is considered thereafter.

\section{Exploiting Sparsity in SDP Relaxation}\label{section-3}

Common approaches for solving SDPs can only handle multiple small PSD matrices efficiently\cite{4-495, 4-496}. In SDP relaxation for DID, the large size of $\bm{Z}$ and $\tilde{\bm{Z}}$ significantly affects computational efficiency of solution approaches for SDPs. Hence this section further exploits sparsity in (P8) to decompose large-dimensional PSD matrix constraints into smaller ones which can be handled much more efficiently.

\subsection{Clique Decomposition Approach}

We will use some basic concepts in graph theory, including maximal clique, chordality, chordal extension and clique tree. Reader should refer to \cite{4-498} for more details of these concepts. The symmetric matrix $\bm{Z}$ can be associated with a graph $\mathcal{G}(\mathcal{V}, \mathcal{E})$ with $\mathcal{V}$ being the row (or column) index set of $\bm{Z}$ and $\mathcal{E} \!\!=\! \{(i,j)| i\!\neq\! j, [\bm{Z}]_{ij}~\text{features in the objective function or}$ $\text{constraints}\}$. Here $\mathcal{E}$ is called the \textit{aggregate sparsity pattern} of SDP in terms of $\bm{Z}$. A matrix $\bm{Z}' \in \mathbb{S}^{|\mathcal{V}|}_{\mathcal{E}}$ is called a symmetric partially specified matrix. A matrix $\bm{Z} \in \mathbb{S}^{|\mathcal{V}|}$ is called a positive semi-definite completion of  $\mathbb{S}^{|\mathcal{V}|}_{\mathcal{E}}$ if $[\bm{Z}]_{ij} = [\bm{Z}']_{ij}$ for all $(i,j) \in \mathcal{E}$ and $\bm{Z} \succeq 0$.

Then the clique decomposition approach to decompose $\bm{Z} \!\!\succeq\!\! 0$ can be summarized as the following steps: \textit{(\romannumeral1)} compute a chordal extension $\mathcal{F} \!\supseteq\! \mathcal{E}$ if $\mathcal{G}$ is not chordal; \textit{(\romannumeral2)} identify the set of maximal cliques $\mathcal{K}\!\!=\!\! \{\mathcal{C}_1,...,\mathcal{C}_{n_{mc}}\}$ of $\mathcal{G}(\mathcal{V}, \mathcal{E})$ or $\mathcal{G}(\mathcal{V}, \mathcal{F})$ if $\mathcal{G}$ is not chordal; \textit{(\romannumeral3)} compute the clique tree $\mathcal{T}(\mathcal{K}, \mathcal{L})$; and \textit{(\romannumeral4)} $\bm{Z} \!\succeq\! 0$ is decomposed into $\mathcal{S}_{\mathcal{C}_i}(\bm{Z}) \!\succeq\! 0, \forall i \!\in\!\! \{1,2,...,n_{mc}\}$ with equality constraints introduced to equate the overlapping entries in maximal cliques \cite{4-495}. 
Considering the complexity of $\tilde{\bm{Z}}$ and $\bm{Z}$, we propose Proposition \ref{prop-chordality-1} (see Appendix\ref{appendix-ppa}) so as to explore sparsity of SDP relaxation in a hierarchical way.

\begin{remark}\label{remark-prop-1}
    Taking $\bm{Z}$ for example, it can be associated with graphs $\mathcal{G}(\mathcal{V}, \mathcal{E})$ and $\bar{\mathcal{G}}(\bar{\mathcal{V}}, \bar{\mathcal{E}})$, at the element and block levels, respectively. Nodes in $\mathcal{V}$ and $\bar{\mathcal{V}}$ correspond to elements and blocks of $\bm{Z}$, respectively. Thus by Proposition \ref{prop-chordality-1}, steps (\romannumeral1) and (\romannumeral2) of the clique decomposition approach can be conducted hierarchically.
\end{remark}

\subsection{Aggregate Sparsity Pattern and Decomposition of SDP}

The results of chordality and maximal cliques of the associated graph of $\bm{Z}$ is given by Proposition \ref{proposition-9} (see Appendix\ref{appendix-ppa}). Denote by $\mathcal{T}(\mathcal{K}, \mathcal{L})$ the clique tree for $\mathcal{G}$. Then $\mathcal{T}$ only needs to satisfy the two properties as follows: 

\textit{(\romannumeral1)} $\forall j \!\!\in\!\! \mathcal{N}_g$, the induced graph by nodes $\{\mathcal{C}_{ej}|\mathcal{C}_{ej} \!\!\in\!\! \bigcup_{k \in \mathcal{D}, i\in \mathbbm{T}^k} \mathcal{K}_{e1}^{ki}\}$ forms a subtree of $\mathcal{T}$; and 

\textit{(\romannumeral2)} $\forall k \in \mathcal{D}, i\in \mathbbm{T}^k$, the induced graph by the node set $\mathcal{K}_{e3}^{ki}$ 
with substitution $\mathcal{C}_{ej} \to \mathcal{C}_{nj} $ forms a clique tree of $\mathcal{G}_n$.\\
Other maximal cliques can be arranged arbitrarily provided that $\mathcal{T}$ is tree-structured. Equality constraints only need to be defined for overlapping entries between maximal cliques within the induced graphs in properties \textit{(\romannumeral1)} and \textit{(\romannumeral2)}.

\begin{remark}\label{remark-asp-1}
    To decompose $\bm{Z} \!\succeq\! 0$, we mainly need to compute a chordal extension of $\mathcal{G}_n$, $\mathcal{K}_n$ and a clique tree of the extended $\mathcal{G}_n$.  In our work, a chordal extension of $\mathcal{G}_n$ is obtained using a fill-reducing Cholesky factorization of matrix $A_{adj} \!+\! \beta_{cf} I$. The Bron-Kerbosch algorithm \cite{4-510} is used to identify $\mathcal{K}_n$. The clique tree is obtained from a maximum-weight spanning tree of a graph with nodes corresponding to $\mathcal{C}_{nj}$ and edge weights between each node pair given by the number of shared buses in each clique pair. The maximum-weight spanning tree can be computed by the Kruskal's algorithm \cite{4-514}. Readers should refer to \cite{4-512, 4-498, 4-510, 4-514} for more details of these computations.
\end{remark}

Regarding the aggregate sparsity pattern of $\tilde{\bm{Z}}$, it is trivial to conclude that $\tilde{\mathcal{G}}$ is chordal with the set of maximal cliques given by
{
\small{
\begin{equation}\label{eq-sparsity-6}\nonumber
        \tilde{\mathcal{K}} \!=\! \bigcup_{k \in \mathcal{D}, i\in \mathbbm{T}^k} \tilde{\mathcal{K}}_i^k ~\text{with}~
\end{equation}
\begin{equation}\label{eq-sparsity-6-1}
\tilde{\mathcal{K}}_i^k \!\!=\!\! \{\mathcal{C}_{ej} | \mathcal{C}_{ej} \!\!=\!\! \tilde{\mathcal{I}}_i^k(\bm{\beta},r,\imath, j) \!\cup \tilde{\mathcal{I}}_i^k(\!\bm{\varsigma},r,\imath, j\!) \!\cup \tilde{\mathcal{I}}(-1), (\!r,\imath, j\!) \!\!\in\!\! \tilde{\mathcal{P}}\}
\end{equation}
}}
where $\tilde{\mathcal{I}}_i^k(\bm{\beta},r,\imath, j)$ represents the index of $\tilde{\bm{Z}}$ corresponding to the $j$th entries in $\bm{\beta}_{(r,\imath)i}^k$, with $(r,\imath, j) \!\in\! \tilde{\mathcal{P}} \!=\! \{0,...,n_c\} \!\times\! \mathcal{B} \!\times\! \{1,2,3\}$, and $\tilde{\mathcal{I}}_i^k(\bm{\varsigma}, r,\imath, j)$ the index of $\tilde{\bm{Z}}$ corresponding to the $j$th entries in $\bm{\varsigma}_{(r,\imath)i}^k$, with $(r,\imath, j) \!\in\! \tilde{\mathcal{P}} \!=\! \{0,...,n_c\} \!\times\! \mathcal{B} \!\times\! \{1,2,3\}$. No equality constraints need to be defined since overlapping entries are all $\mathcal{S}_{\tilde{\mathcal{I}}(-1)}(\tilde{\bm{Z}}) = 1$.

\subsection{Decomposition of the SDP Relaxation}

We further introduce symmetric matrix variables $\hat{\bm{Z}}_{\mathcal{C}_{ej}} \in \mathbb{R}^{|\mathcal{C}_{ej}| \times |\mathcal{C}_{ej}|}$, $\forall \mathcal{C}_{ej} \in \mathcal{K} \cup \tilde{\mathcal{K}}$. According to step \textit{(\romannumeral4)} of the clique decomposition approach, (P8) can be decomposed as
{
\footnotesize{

\begin{subequations}\label{eq-sparsity-7}
	\begin{align}
        &(\text{P}9)~{\min_{  \{ \hat{\bm{Z}}_{\mathcal{C}_{ej}} \} }} 
        \sum_{k \in \mathcal{D}, i \in \mathbb{T}^k} \sum_{ \mathcal{C}_{ej} \in \mathcal{K}_i^k} \Tr(\bm{P}_{\mathcal{C}_{ej}}  \hat{\bm{Z}}_{\mathcal{C}_{ej}} ) \label{eq-sparsity-7:1}\\
    &\mathrm{s.t.}_{\forall k \in \mathcal{D}, i \in \mathbb{T}^k} \nonumber\\
    ~&  \sum_{ \mathcal{C}_{ej} \in \mathcal{K}_i^k \cup \tilde{\mathcal{K}}_i^k} \Tr(\bm{P}^{\gamma}_{\mathcal{C}_{ej}}  \hat{\bm{Z}}_{\mathcal{C}_{ej}} ) \leq 0 ~~ \gamma = 1, 2,...,n_{\gamma}    \label{eq-sparsity-7:2} \\
    ~& \mathcal{S}_{I_j} (\hat{\bm{Z}}_{\mathcal{C}_{ej}}) \!-\! \mathcal{S}_{I'_j}  (\hat{\bm{Z}}_{\mathcal{C}'_{ej}}) \!\!=\! \bm{0} ~~ \forall \mathcal{C}'_{ej} \in \text{ch}( \mathcal{C}_{ej} ) \!\cap\! \mathcal{K}_{e3}^{ki}, \mathcal{C}_{ej} \!\in\! \mathcal{K}_{e3}^{ki} \label{eq-sparsity-7:3} \\
    ~&  \mathcal{S}_{I_j} \!(\!\hat{\bm{Z}}_{\mathcal{C}_{ej}}\!) \!-\! \mathcal{S}_{I'_j} (\hat{\bm{Z}}_{\mathcal{C}'_{ej}}) \!\!=\! \bm{0} ~~~~ \forall \mathcal{C}'_{ej} \!\!\in\!\! \text{ch}( \mathcal{C}_{ej} ) \!\!\cap\! \mathcal{K}^j \!,\! \mathcal{C}_{ej} \!\!\in\!\! \mathcal{K}_{e1}^{ki}, j \!\!\in\!\! \mathcal{N}_g     \label{eq-sparsity-7:4} \\
    ~&  \mathcal{S}_{I_\theta} (\hat{\bm{Z}}_{\mathcal{C}_{ej}}) \!\!-\!\! \mathcal{S}_{I'_\theta} (\hat{\bm{Z}}_{\mathcal{C}'_{ej}}) \!\!=\! \bm{0} ~~~~  \forall \mathcal{C}_{nj} \!\!\in\!\! \mathcal{K}_n,  \mathcal{C}_{ej} \!\!\in\!\! \mathcal{K}_{e3}^{ki},  \mathcal{C}'_{ej} \!\!\in\!\! \mathcal{K}_{e3}^{k(i-1)} \label{eq-sparsity-7:5} \\
    ~&  \mathcal{S}_{I_\omega} (\hat{\bm{Z}}_{\mathcal{C}_{ej}}) \!\!-\!\! \mathcal{S}_{I'_\omega} (\hat{\bm{Z}}_{\mathcal{C}'_{ej}}) \!\!=\!\! \bm{0} ~~~~ \forall j \!\in\! \mathcal{N}_g , \mathcal{C}_{ej} \!\!\in\! \mathcal{K}_{e1}^{ki}, \mathcal{C}'_{ej} \!\!\in\! \mathcal{K}_{e1}^{k(i-1)}  \label{eq-sparsity-7:6} \\
    ~& \hat{\bm{Z}}_{\mathcal{C}_{ej}} \succeq 0~~ \forall \mathcal{C}_{ej} \in \mathcal{K}_i^k \cup \tilde{\mathcal{K}}_i^k \label{eq-sparsity-7:7}
 	\end{align}
\end{subequations}
}}
Here in (\ref{eq-sparsity-7:4}), all intersections between cliques for any given $j \!\in\! \mathcal{N}_g$ correspond to the same principal submatrix of $\bm{Z}$ defined by index set $\{\mathcal{I}(M\mathbbm{1},j), \mathcal{I}(D\mathbbm{1},j), \mathcal{I}(-1)\}$; in (\ref{eq-sparsity-7:5}), $\mathcal{C}_{ej}$ and $\mathcal{C}'_{ej}$ are both corresponding to $\mathcal{C}_{nj}$; (\ref{eq-sparsity-7:2}) are equivalent to (\ref{eq-sdp-1:2}), (\ref{eq-sdp-1:4}), (\ref{eq-sdp-4:2}) to (\ref{eq-sdp-4:4}) and (\ref{eq-sdp-4:5-1}) to (\ref{eq-sdp-4:5-4}), respectively; (\ref{eq-sparsity-7:3}) and (\ref{eq-sparsity-7:4}) equate the overlapping entries; (\ref{eq-sparsity-7:5}) and (\ref{eq-sparsity-7:6}) are equivalent to the second equations in (\ref{eq-sdp-1:4}) and (\ref{eq-sdp-4:4}); and (\ref{eq-sparsity-7:7}) is equivalent to (\ref{eq-sdp-4p:4}). Thereby, large size PSD matrices in (P8), i.e., $\bm{Z}$ and $\tilde{\bm{Z}}$, are decomposed into multiple much smaller ones, i.e., $\hat{\bm{Z}}_{\mathcal{C}_{ej}}$ in (\ref{eq-sparsity-7:7}), and (P9) can be solved more efficiently than (P8).

\section{ADMM-Based Distributed Optimization for DID}\label{section-4}

The collocation method endows the DID formulation flexibility to handle various forms of disturbances and bound constraints, while an inevitable consequence is high dimension of (P2) and induced problems. Moreover, the SDP relaxation could be inexact. In this section, we propose a feasibility-embedded distributed approach for solving (P9) to address these two issues simultaneously under the framework of ADMM. For the high dimension, ADMM is utilized to separate the SDP with intractable size into a series of small-size SDPs which can be solved in parallel. To cope with inexactness of the SDP relaxation, we further embed solution feasibility into the solution process with the idea of the ADMM-based restricted low-rank approximation approach \cite{4-507}. In this way, near-globally optimal solutions of (P2) are promisingly obtained with moderate computational burdens.

% \noindent
% \textit{1) Matrix vector and its operations}

\textit{Notaiton}. To simplify expressions, a special kind of variables and corresponding operations are first introduced. We call $\mathcal{X} \!\!=\!\! [\bm{X}_1, \bm{X}_2, ..., \bm{X}_n]^T$ where $\bm{X}_i$ are all matrices, a \textit{matrix vector}. If another matrix vector $\mathcal{Y} \!=\! [\bm{Y}_1, \bm{Y}_2, ..., \bm{Y}_n]^T$ satisfies that $\forall i \!\!=\!\! 1,2,...,n$, $\bm{X}_i$ and $\bm{Y}_i$ are with the same size, then $\mathcal{X}$ and $\mathcal{Y}$ are with the same size; $\mathcal{X} \!\pm\! \mathcal{Y} \!=\! [\bm{X}_1 \!\pm\! \bm{Y}_1,..., \bm{X}_n \!\pm\! \bm{Y}_n]^T$; $\mathcal{X}^{T^*} \!\circ \mathcal{Y} \!\!=\!\! [\bm{X}_1^T \bm{Y}_1, \bm{X}_2^T \bm{Y}_2,..., \bm{X}_n^T \bm{Y}_n]^T$; $\Tr(\mathcal{X}) \!\!=\!\! \sum_{i=1}^n \Tr(\bm{X}_i)$; and $\mathcal{X}^{T^*} \!\!=\!\!  [\bm{X}_1^T, \bm{X}_2^T, ..., \bm{X}_n^T]^T$. Frobenius norm of $\mathcal{X}$ is define as $\Vert \mathcal{X} \Vert_F \!\!=\!\! \sqrt{ \!\sum_{i=1}^n \!\Tr(\bm{X}_i^T \bm{X}_i) }$. A \textit{linear matrix vector function} is defined as $\mathcal{X} \!\!=\!\! [\bm{X}_1, \!\bm{X}_2, ..., \!\bm{X}_n]^T \!\!\to\!\! \mathcal{Z} \!=\! [\bm{Z}_1, \bm{Z}_2, ..., \bm{Z}_m]^T$, where $\mathcal{X}$ and $\mathcal{Z}$ are matrix vectors, and $\forall j\!=\! 1,2,...,m$, $\exists i \!=\! 1,2,...,n$ and matrix $\bm{P}_j$, let $\bm{Z}_j \!\!=\!\! \bm{P}_j \bm{X}_i \bm{P}_j^T$. $\rank(\mathcal{X}) \!\!=\!\! \max\{\rank(\bm{X}_i)| i\!=\!1,2,..,n\}$; $\vect(\mathcal{X}) \!=\! [\vect(\bm{X}_1)^T,$ $\vect(\bm{X}_2)^T,...., \vect(\bm{X}_n)^T]^T$ with $\vect(\bm{X}_i)$ being the vectorization of matrix $\bm{X}_i$; $\mathcal{X}$ is called a symmetric matrix vector or $\mathcal{X}$ is symmetric if $\forall i \!=\! 1,2,...,n$, $\bm{X}_i$ is a symmetric matrix; $\diag(\mathcal{X}) \!\!=\!\! [\diag(\bm{X}_1)^T, \diag(\bm{X}_2)^T,...,\diag(\bm{X_n})^T]^T$; $\text{lower}(\mathcal{X})\!\!=\!\! [\text{lower}(\bm{X}_1)^T\!, \text{lower}(\bm{X}_2)^T\!,...,\text{lower}(\bm{X}_n)^T]^T$ with $\text{lower}(\bm{X}_i)^T$ being the vectorization of all the entries under the main diagonal in matrix $\bm{X}_i$ following a column-major order; and $\text{upper}(\mathcal{X})$ is analogous to $\text{lower}(\mathcal{X})$ but for all the entries above the main diagonal and following a row-major order.

\subsection{Distributed Solution Approach}\label{sec-compute-dist}
We first present the distributed approach for solving (P9) without considering feasibility of solutions. Separation is conducted by decoupling portions corresponding to different sets of disturbances and time elements. Without loss of generality, it is assumed that all pairs of $(k,i)$ are divided into $N_S$ disjoint sets, i.e., $\Xi_s = \{...,(k,i),...\}$ with $s \in \mathbb{P}= \{1,...,N_S\}$ and $\forall s \in \mathbb{P}$, define the following sets:
\begin{subequations}\label{eq-6-20-1}
    \begin{align}
        & \overline{\mathcal{M}}_s \!\!=\!\! \{\!(k, i_k) | (k, i_k) \!\!\in\! \Xi_s \!\wedge\! i_k \!\!\in\!\! \text{max}^\circ \!\{i|(k, i) \!\!\in\!\! \Xi_s\!\} \!\backslash\! \{n_t^k\} \!\}\\
        & \underline{\mathcal{M}}_s \!\!=\!\! \{\!(k, i_k) | (k, i_k) \!\!\in\!\! \Xi_s \!\wedge\! i_k \!\!\in\!\! \text{min}^\circ \!\{i|(k, i) \!\in\! \Xi_s\} \!\backslash\! \{1\} \}\\
        & \tilde{\mathcal{M}}_s = \Im(\Xi_s)
    \end{align}
\end{subequations}
To make (P9) separable regarding variables corresponding to each $\Xi_s$, the following auxiliary matrix variables are created:
\begin{subequations}\label{eq-6-21}
    \begin{align}
        & \bm{Z}_{j}^{\text{md}} \in \mathbb{R}^3  & & \forall j \in \mathcal{N}_g \label{eq-6-21:1}\\
        & \bm{Z}_{\mathcal{C}_{ej}}^{\theta }  \in \mathbb{R}^{ |\mathcal{C}_{nj}| }  & &
        \forall (k,i) \in \underline{\mathcal{M}}, \mathcal{C}_{nj} \!\!\in\!\! \mathcal{K}_n, \mathcal{C}_{nj} \!\!\to\!\! \mathcal{C}_{ej} \!\!\in\!\! \mathcal{K}_{e3}^{ki} \\
        &  \bm{Z}_{\mathcal{C}_{ej}}^{\omega }  \in \mathbb{R}^{ 2 }  & & 
        \forall (k,i) \in \underline{\mathcal{M}}, j \in \mathcal{N}_g , \mathcal{C}_{ej} \in \mathcal{K}_{e1}^{ki} 
        %& [ \tilde{{\theta}}_{k,i}, \tilde{{\omega}}_{k,i}] = [\bm{L}_2 \bm{\theta}_{(i-1)}^{k}, \bm{L}_4 \bm{\omega}_{(i-1)}^{k}],  && \forall (k,i) \in \underline{\mathcal{M}} \label{eq-6-21:2}
   \end{align}
\end{subequations} 
Then we divide constraints in (P9) into two groups, called inner constraints that involve only variables corresponding to cliques from the same $\mathcal{K}_i^k$, and coupling constraints that involve variables corresponding to cliques from different $\mathcal{K}_i^k$, respectively. Specifically, constraints (\ref{eq-sparsity-7:2}), (\ref{eq-sparsity-7:3}) and (\ref{eq-sparsity-7:7}) are inner constraints and the others are coupling constraints. Furthermore, (P9) can be rewritten as the following compact form:
\begin{subequations}\label{eq-6-23}
	\begin{align}
		(\text{P}10)~  {\min_{ \hat{\mathcal{Z}}_{\!s} \!\in \mathbb{Z}_{\!s}, \mathcal{Z}_{\text{a}}  \!\in \mathbb{R}_\text{a} }}  & \sum_{s \in \mathbb{P}}\! {\hat{J}_{\!s}(\! \hat{\mathcal{Z}}_{\!s} \!)} \!\!=\!\!\! \sum_{\!\!\!s \in \mathbb{P}} \! \sum_{\!(\!k,i\!) \!\in \Xi_{\!s} } \! \sum_{ \mathcal{C}_{\!e\!j} \!\in \mathcal{K}_i^k} \!\!\! \Tr(\!\bm{P}_{\mathcal{\!C}_{\!e\!j}} \! \hat{\bm{Z}}_{\!\mathcal{C}_{\!e\!j}} \!) \label{eq-6-23:1}\\ 
    \mathrm{s.t.}_{\forall s \in \mathbb{P}}
    ~&  \zeta_s^A(\hat{\mathcal{Z}}_s) -  \zeta_s^B( \mathcal{Z}_{\text{a}} )  = \mathcal{O}_s \label{eq-6-23:2}
	\end{align}
\end{subequations}
where $\zeta_s^A$ and $\zeta_s^B$ are proper linear matrix vector functions to make constraint (\ref{eq-6-23:2}) equivalent to the following ones:
{\footnotesize
\begin{subequations}\label{eq-6-25}
	\!\!\begin{align}
    ~&  \mathcal{S}_{I_j} (\hat{\bm{Z}}_{\mathcal{C}_{ej}}) \!-\! \bm{Z}_{j}^{\text{md}} \!=\! \bm{0} ~~ \forall j \!\in\! \mathcal{N}_g, \mathcal{C}_{ej} \!\!\in\! \mathcal{K}_{e1}^{ki}, (k,i) \!\!\in\! \tilde{\mathcal{M}}_s \label{eq-6-25:0}\\
    ~&  \mathcal{S}_{I_\theta} (\hat{\bm{Z}}_{\mathcal{C}_{ej}}) \!\!-\!\!  \bm{Z}_{\mathcal{C}_{ej}}^{\theta }  \!\!=\!\!  \bm{0} ~~  \forall \mathcal{C}_{nj} \!\!\in\!\! \mathcal{K}_n, \mathcal{C}_{ej} \!\!\in\!\! \mathcal{K}_{e3}^{ki}, (k,i) \!\in\! \underline{\mathcal{M}}_s \label{eq-6-25:1}\\
    ~&  \mathcal{S}_{I'_\theta} \!(\!\hat{\bm{Z}}_{\mathcal{C}'_{ej}}\!)\!  \!\!-\!\!  \bm{Z}_{\mathcal{C}_{ej}}^{\theta }  \!\!\!=\!  \bm{0} ~~\forall \mathcal{C}_{nj} \!\!\in\!\! \mathcal{K}_n,\! \mathcal{C}'_{\!ej} \!\!\in\!\! \mathcal{K}_{e3}^{k(\!i\!-\!1\!)}\!\!, \mathcal{C}_{ej} \!\!\in\!\! \mathcal{K}_{e3}^{ki}, \!(\!k,\!i\!\!-\!\!1) \!\!\in\!\! \overline{\mathcal{M}}_{\!s}  \label{eq-6-25:2}\\
    ~&  \mathcal{S}_{I_\omega} (\hat{\bm{Z}}_{\mathcal{C}_{ej}}) \!\!-\!\! \bm{Z}_{\mathcal{C}_{ej}}^{\omega } \!\!=\! \bm{0} ~~ \forall j \!\in\! \mathcal{N}_g , \mathcal{C}_{ej} \!\in\! \mathcal{K}_{e1}^{ki}, (k,i) \!\in\! \underline{\mathcal{M}}_s \\
    ~&  \mathcal{S}_{\!I'_\omega} \!(\!\hat{\bm{Z}}_{\mathcal{C}'_{\!ej}}\!)\! \!-\!\! \bm{Z}_{\mathcal{C}_{\!ej}}^{\omega } \!\!\!=\!\! \bm{0}
    ~~\forall j \!\!\in\!\ \mathcal{N}_g,\! \mathcal{C}'_{\!ej} \!\!\in\! \mathcal{K}_{e1}^{k\!(\!i\!-\!1\!)}\!\!, \mathcal{C}_{ej} \!\!\in\!\! \mathcal{K}_{e1}^{ki}, \!(\!k,\!i\!-\!1\!) \!\in\! \overline{\mathcal{M}}_s 
	\end{align} 
\end{subequations}
}
with $\mathcal{C}_{ej}$ and $\mathcal{C}'_{ej}$ corresponding to $\mathcal{C}_{nj}$ in (\ref{eq-6-25:1}) and (\ref{eq-6-25:2}).

Clearly, (P9) can be regarded as a consensus problem with both global and local variables. For any given $j \!\in\! \mathcal{N}_g$, all $\mathcal{S}_{I_j} (\hat{\bm{Z}}_{\mathcal{C}_{ej}})$ with $\mathcal{C}_{ej} \!\in\! \mathcal{K}_{e1}^{ki}, (k,i) \!\in\! \tilde{\mathcal{M}}_s$ and $s \!\in\! \mathbb{P}$, achieve consensus globally; for any given  $(k,i) \!\in\! \underline{\mathcal{M}}_s$, $ \mathcal{S}_{I_\theta} (\hat{\bm{Z}}_{\mathcal{C}_{ej}})$ and $\mathcal{S}_{I'_\theta} (\hat{\bm{Z}}_{\mathcal{C}'_{ej}})$, with $\mathcal{C}_{nj} \!\!\in\!\! \mathcal{K}_n, \mathcal{C}_{nj} \!\!\to\!\! \mathcal{C}_{ej} \!\!\in\!\! \mathcal{K}_{e3}^{ki}$ and $\mathcal{C}_{nj} \!\!\to\!\! \mathcal{C}'_{ej} \!\!\in\!\! \mathcal{K}_{e3}^{k(i-1)}$, achieve consensus locally; and for any given  $(k,i) \in \underline{\mathcal{M}}_s$ and $j \in \mathcal{N}_g$, $ \mathcal{S}_{I_\omega} (\hat{\bm{Z}}_{\mathcal{C}_{ej}})$ and $\mathcal{S}_{I'_\omega} (\hat{\bm{Z}}_{\mathcal{C}'_{ej}})$, with $ \mathcal{C}_{ej} \in \mathcal{K}_{e1}^{ki}$ and $\mathcal{C}'_{ej} \in \mathcal{K}_{e1}^{k(i-1)}$, achieve consensus locally. According to the consensus ADMM \cite{4-461}, iterations for solving (P10) are expressed as
{\small{
\begin{subequations}\label{eq-6-28}
    \begin{align}
        & \hat{\mathcal{Z}}_s^{(\kappa + 1)} \!:=\! \argmin_{ \hat{\mathcal{Z}}_s \in \mathbb{Z}_s} 
        \!\left(\!\!\!
        \begin{aligned}
            \hat{J}_s( \hat{\mathcal{Z}}_s ) + \Tr( \mathcal{A}_s^{(\kappa)T^*} \!\!\!\! \circ\! \zeta_s^A(\hat{\mathcal{Z}}_s) ) \\
            + \frac{\rho}{2}\Vert \zeta_s^A(\hat{\mathcal{Z}}_s) \!-\!  \zeta_s^B( \mathcal{Z}_{\text{a}}^{(\kappa)} ) \Vert_F^2
        \end{aligned}
        \!\right)\!~~ \forall s \!\in\! \mathbb{P} \label{eq-6-28:1}\\
        & \mathcal{Z}_{\text{a}}^{(\kappa + 1)} \!\!:= \nonumber \\
        & \begin{tikzpicture}
            \matrix(B)[matrix of math nodes, row sep={0.9cm,between origins}, inner sep=-0.2pt, nodes={node style ge}, left delimiter={[}, right delimiter={]}] at (0,0){
                \!\!\bm{Z}_{j}^{\text{md} (\kappa + 1) }   \!\!\\ 
                \!\!\bm{Z}_{\mathcal{C}_{ej}}^{\theta (\kappa + 1) }  \!\!\\
                \!\!\bm{Z}_{\mathcal{C}_{ej}}^{\omega (\kappa + 1) } \!\! \\
            };
            \matrix(C)[matrix of math nodes, nodes={node style ge}] at (1,-0.15){
                \!\!\!\!\!\!=\! \\
            };
            \matrix(D)[matrix of math nodes, row sep={0.9cm,between origins}, inner sep=-0.2pt, nodes={node style ge}, left delimiter={[}, right delimiter={]}] at (3.1,0){
                \frac{1}{|\mathbb{P}|} \sum_{ (k,i) \in \tilde{\mathcal{M}}  } \mathcal{S}_{I_j} (\hat{\bm{Z}}_{\mathcal{C}_{ej}}^{(\kappa + 1)}  )   \\
                \!\!\! \frac{1}{2} \!\!\left(\!\mathcal{S}_{I_\theta} \!(\!\hat{\bm{Z}}_{\mathcal{C}_{ej}}^{(\kappa + 1)} ) \!\!+\!\!   \mathcal{S}_{I'_\theta} (\hat{\bm{Z}}_{\mathcal{C}'_{ej}}^{(\kappa + 1)} ) \! \right) \! \\
                \!\! \frac{1}{2} \!\!\left(\!\mathcal{S}_{I_\omega} \!(\!\hat{\bm{Z}}_{\mathcal{C}_{ej}}^{(\kappa + 1)} ) \!\!+\!\!   \mathcal{S}_{I'_\omega} (\hat{\bm{Z}}_{\mathcal{C}'_{ej}}^{(\kappa + 1)} )  \!\right) \! \\
            };
            \matrix(E)[matrix of math nodes, row sep={0.9cm,between origins}, inner sep=0.2pt, nodes={node style ge}] at (6.1,0){
                \forall j \!\!\in\!\! \mathcal{N}_g  \\
                \begin{aligned}
                    \forall \mathcal{C}_{nj} \!\!\in\!\! \mathcal{K}_n, \\
                    (k,i) \!\!\in\!\! \underline{\mathcal{M}}
                \end{aligned}\\
                \begin{aligned}
                    \forall j\in \mathcal{N}_g, \\
                    (k,i) \!\!\in\!\! \underline{\mathcal{M}}
                \end{aligned} \\
            };  
            \draw [dashed] (-0.6,0.525) to (6.9,0.525);
            \draw [dashed] (-0.6,-0.383) to (6.9,-0.383);
            % \draw [BarreStyle=red!40, line width=16 mm] (-0.8,1.42) to (8,1.42) ;
            % \draw [BarreStyle=green!50] (-0.8,-0.1) to (8,-0.1) ;
            % \draw [BarreStyle=blue!50] (-0.8,-1.52) to (8,-1.52) ;
        \end{tikzpicture} \label{eq-6-28:2}\\
        &  \mathcal{A}_s^{(\kappa + 1)} \!:=\! \mathcal{A}_s^{(\kappa)} \!\!+\! \rho\! \left(\! \zeta_s^A(\hat{\mathcal{Z}}_s^{(\kappa + 1)}) \!-\!  \zeta_s^B(\! \mathcal{Z}_{\text{a}}^{(\kappa + 1)} \!)  \!\right) ~ \forall s \!\in\! \mathbb{P} .\label{eq-6-28:3}
   \end{align}
\end{subequations}
}}
$\!\!\!$where $\mathcal{A}_s$ is the matrix vector of dual variables associated with constraints (\ref{eq-6-23:2}) and with the same size as $\zeta_s^A(\hat{\mathcal{Z}}_s)$; in (\ref{eq-6-28:2}), for the upper block, $\mathcal{C}_{ej} \in \mathcal{K}_{e1}^{ki}$, for the middle block, $\mathcal{C}_{nj} \!\!\to\!\! \mathcal{C}_{ej} \!\!\in\!\! \mathcal{K}_{e3}^{ki}$ and $\mathcal{C}_{nj} \!\!\to\!\! \mathcal{C}'_{ej} \!\!\in\!\! \mathcal{K}_{e3}^{k(i-1)}$, and for the lower block, $\mathcal{C}_{ej} \in \mathcal{K}_{e1}^{ki}$ and $\mathcal{C}'_{ej} \in \mathcal{K}_{e1}^{k(i-1)}$. The objective function in (\ref{eq-6-28:1}) is the augmented Lagrangian of (P10), and (\ref{eq-6-28:2}) and (\ref{eq-6-28:3}) update auxiliary variables and dual variables, respectively. We refer the readers to \cite{4-461} for detailed derivations of (\ref{eq-6-28}).

The Frobenius norm in (\ref{eq-6-28:1}) makes updates of primal variables not SDPs. However, by introducing slack variables, (\ref{eq-6-28:1}) is equivalent to
\begin{equation}\label{eq-6-28-1}
    \{ \hat{\mathcal{Z}}_s^{(\kappa + 1)}, \cdot \} \!:=\! 
        \argmin_{ \{ \hat{\mathcal{Z}}_s, \varphi_s \} \in \mathbb{Z}_s \cap \mathbb{Z}_s^{\phi} } 
           \hat{J}_s( \hat{\mathcal{Z}}_s ) + \frac{\rho}{2} \varphi_s  ~~~ \forall s \in \mathbb{P}
\end{equation}
with feasible region $\mathbb{Z}_s^{\phi}$ defined by
\begin{equation}\label{eq-6-28-2}
    \left[\!\!
    \begin{array}{cc}
        \varphi_s & \vect \left( \zeta_s^A(\hat{\mathcal{Z}}_s) \!\!-\!\! ( \zeta_s^B( \mathcal{Z}_{\text{a}}^{(\kappa)} ) \!\!-\!\! \frac{1}{\rho} \mathcal{A}_s^{(\kappa)} ) \right)^T        \\
        * & I \\   
    \end{array}
    \!\! \right]  \!\! \succeq 0.
\end{equation}
Here and hereafter, "$*$" is used to denote partial entries of symmetric matrices. Therefore updates of primal variables can still be conducted by solving SDPs.

\begin{remark}\label{remark-sdp-reduction}
    The size of SDP constraint (\ref{eq-6-28-2}) can be reduced by considering the symmetry of $\zeta_s^A(\hat{\mathcal{Z}}_s) \!-\! ( \zeta_s^B( \mathcal{Z}_{\mathrm{a}}^{(\kappa)} ) \!-\! \frac{1}{\rho} \mathcal{A}_s^{(\kappa)} )$ and introducing multiple slack variables. For simplicity, we use $\zeta$ to represent $\zeta_s^A(\hat{\mathcal{Z}}_s) \!-\! ( \zeta_s^B( \mathcal{Z}_{\mathrm{a}}^{(\kappa)} ) \!-\! \frac{1}{\rho} \mathcal{A}_s^{(\kappa)} )$ in this remark only. In constraint (\ref{eq-6-28-2}), $\zeta_s^A(\hat{\mathcal{Z}}_s)$ and $\zeta_s^B( \mathcal{Z}_{\mathrm{a}}^{(\kappa)})$ are symmetric for all $\kappa \geq 0$. Thus $\zeta$ is symmetric as long as $\mathcal{A}_s^{(\kappa)}$ is symmetric, which can be guaranteed by setting $\mathcal{A}_s^{(0)}$ to a symmetric matrix vector. Furthermore, $\mathrm{upper}(\zeta) = \mathrm{lower}(\zeta)$ with $\zeta$ being symmetric. Therefore, if $\mathcal{A}_s^{(0)}$ is symmetric, constraint (\ref{eq-6-28-2}) can be replaced by the following equivalent form:
    \begin{equation}\label{eq-remark-2-1}
        \left\{
        \begin{aligned}
            & \varphi_s = 2 \sum_{i=1}^{n_\mathrm{u}}  \varphi_{s,i}^{\mathrm{u}} + \sum_{i=1}^{n_\mathrm{d}} \varphi_{s,i}^{\mathrm{d}}\\
            & \left[\!\!\!\!
             \begin{array}{cc}
                \varphi_{s,i}^{\mathrm{d}} & \diag(\zeta)^T_i  \\
                * & I 
            \end{array}
            \!\!\!\! \right]  \!\! \succeq 0,  ~~~~ i = 1,2,...,n_\mathrm{d}\\
            & \left[\!\!\!\!
            \begin{array}{cc}
               \varphi_{s,i}^{\mathrm{u}} & \mathrm{upper}(\zeta)^T_i  \\
               * & I 
           \end{array}
           \!\!\!\! \right]  \!\! \succeq 0,  ~~~~ i = 1,2,...,n_\mathrm{u}
        \end{aligned}
        \right.
    \end{equation}
\end{remark} 

% \begin{remark}
% A special separation case is that $\forall k \!\in\! \mathcal{D}$ and $s \in \mathbb{P}$, $\{i|(k,i) \!\in\! \Xi_s\} \!=\! \mathbb{T}^k$ or $\emptyset$, i.e., all time elements for any one disturbance are in the same $\Xi_s$. In this case, $\forall s \!\in\! \mathbb{P}$, $\underline{\mathcal{M}} \!=\! \underline{\mathcal{M}}_s \!=\! \overline{\mathcal{M}}_s \!=\! \emptyset $, and in (P10), $\zeta_s^A$ and $\zeta_s^B$ are proper linear matrix vector functions to make (\ref{eq-6-23:2}) equivalent to (\ref{eq-6-25:1}). Thus (P10) becomes a global consensus problem, and iterations for solving it can be simplified as: 
% \begin{subequations}\label{eq-6-29}
%     \begin{align}
%         & \hat{\mathcal{Z}}_s^{(\kappa + 1)} \!:=\!  
%         \argmin_{ \hat{\mathcal{Z}}_s \in \mathbb{Z}_s} 
%         \!\left(\!
%         \begin{aligned}
%             & \hat{J}_s( \hat{\mathcal{Z}}_s ) \!+\! \Tr( \mathcal{A}_s^{(\kappa)T^*} \!\circ \zeta_s^A(\hat{\mathcal{Z}}_s) )+ \\
%             & \frac{\rho}{2} \Vert \zeta_s^A(\hat{\mathcal{Z}}_s) \!-\! \frac{1}{|\mathbb{P}|}\!\! \sum_{s \in \mathbb{P}} \zeta_s^A(\hat{\mathcal{Z}}_s^{(\kappa)})   \Vert_F^2
%         \end{aligned}
%         \!\!\right)\! \\
%         & \mathcal{A}_s^{(\kappa + 1)} := \!\mathcal{A}_s^{(\kappa)} \!\!+\!\! \rho \!\! \left(\!\! \zeta_s^A(\hat{\mathcal{Z}}_s^{(\kappa + 1)}) \!-\!  \frac{1}{|\mathbb{P}|} \!\! \sum_{s \in \mathbb{P}}\! \zeta_s^A(\hat{\mathcal{Z}}_s^{(\kappa + 1)})  \!\!\right).
%     \end{align}
% \end{subequations}
% \end{remark}

\subsection{Feasibility-Embedded Distributed Solution Approach}
Considering the SDP relaxation before being decomposed, i.e., (P8), when the relaxation is inexact, the optimal solution of (P8) is infeasible to (P7) and thus cannot be used for DID. This feasibility is determined by whether the solution satisfies rank constraints (\ref{eq-sdp-4p:5}). In the following, we further embed feasibility of solutions into the above solving approach by taking the rank constraints into account.

Corollary \ref{corollary-1} (see Appendix\ref{appendix-ppa}) gives the condition for principal submatrices in which the symmetric partially specified matrix has a completion that is not only positive semi-definite but also rank-1. With Corollary \ref{corollary-1} and the fact that $\hat{\bm{Z}}_{\mathcal{C}_{ej}}$ contains at least one non-zero entry, rank constraints (\ref{eq-sdp-4p:5}) are decomposed into
\begin{equation}\label{eq-sparsity-8}
    \text{rank}( \hat{\bm{Z}}_{\mathcal{C}_{ej}}  ) = 1 ~~~ \forall \mathcal{C}_{ej} \in \mathcal{K} \cup \tilde{\mathcal{K}}
\end{equation}

Now we consider (P10) with rank constraints. First $\forall s \in \mathbb{P}$, introduce an auxiliary matrix vector variable $\mathcal{Y}_s \!=\![\bm{Y}_{\mathcal{C}_{ej}}]^T$ with $\mathcal{C}_{ej} \!\in\! \mathcal{K}_s \!=\! \bigcup_{(k,i) \in \Xi_s} ( \mathcal{K}_i^k \!\cup\! \tilde{\mathcal{K}}_i^k )$, and $\mathcal{Y}_s$ and $\hat{\mathcal{Z}}_s$ are of the same size. Then (P10) with rank constraints can be formulated as
\begin{subequations}\label{eq-6-30}
	\begin{align}
		(\text{P}11)~  {\min_{ \hat{\mathcal{Z}}_s \in \mathbb{Z}_{s}, \mathcal{Z}_{\text{a}}  \in \mathbb{R}_\text{a}, \mathcal{Y}_s \in \mathbb{Y}_s }}  & \sum_{s \in \mathbb{P}} {\hat{J}_{s}( \hat{\mathcal{Z}}_s )} \label{eq-6-30:1}\\ 
    \mathrm{s.t.}_{\forall s \in \mathbb{P}}
    ~&  \zeta_s^A(\hat{\mathcal{Z}}_s) -  \zeta_s^B( \mathcal{Z}_{\text{a}} )  = \mathcal{O}_s \label{eq-6-30:2}\\
    ~&  \hat{\mathcal{Z}}_s -  \mathcal{Y}_s = \mathcal{O}'_s \label{eq-6-30:3}
	\end{align}
\end{subequations}   
% \begin{equation}\label{eq-6-31}
%     \rank( \mathcal{Y}_s)  \leq 1 
% \end{equation}
The augmented (partial) Lagrangian of (\ref{eq-6-30}) is written as
\begin{equation}\label{eq-6-32}
        L = \sum_{s \in \mathbb{P}} L_s(\hat{\mathcal{Z}}_s, \mathcal{Z}_{\text{a}}, \mathcal{Y}_s, \mathcal{A}_s, \tilde{\mathcal{A}_s} )
\end{equation}
with
\begin{equation}\label{eq-6-32-1}
    \!\!\!\!\begin{aligned}
        & L_s(\hat{\mathcal{Z}}_s, \mathcal{Z}_{\text{a}}, \mathcal{Y}_s, \mathcal{A}_s, \tilde{\mathcal{A}_s} )  =  \hat{J}_{s}( \hat{\mathcal{Z}}_s )\\
        &  +\!  \Tr \!\left( \mathcal{A}_s^{T^*} \!\!\!\circ\! [\zeta_s^A(\hat{\mathcal{Z}}_s) \!-\!  \zeta_s^B( \mathcal{Z}_{\text{a}} ) ]\! \right) \!\!+\! \frac{\rho}{2} \Vert \zeta_s^A(\hat{\mathcal{Z}}_s) \!-\!  \zeta_s^B( \mathcal{Z}_{\text{a}} ) \Vert_F^2  \\
        &  +\! \Tr \!\left( \tilde{\mathcal{A}}_s^{T^*} \!\!\circ\! ( \hat{\mathcal{Z}}_s -  \mathcal{Y}_s ) \right)  +  \frac{\tilde{\rho}}{2} \Vert \hat{\mathcal{Z}}_s - \mathcal{Y}_s \Vert_F^2, 
    \end{aligned}
\end{equation}
%, and the value of $\tilde{\rho}$ can be different from that of $\rho$. 

Furthermore, iterations for solving (P11) are given by \textit{Step 1)} to \textit{Step 3)} as follows:

\noindent
\textit{Step 1) Update primal variables} 

The update of primal variables $\hat{\mathcal{Z}}_s$ is given as
\begin{equation}\label{eq-6-33-1}
    \hat{\mathcal{Z}}_s^{(\kappa + 1)} \!:=\! 
        \argmin_{ \hat{\mathcal{Z}}_s \in \mathbb{Z}_s}  
        \!L_s(\!\hat{\mathcal{Z}}_s, \!\mathcal{Z}_{\text{a}}^{(\kappa)}, \!\mathcal{Y}_s^{(\kappa)}\!,\! \mathcal{A}_s^{(\kappa)}\!,\! \tilde{\mathcal{A}}_s^{(\kappa)} ) ~~  \forall s \in \mathbb{P}
\end{equation}
Analogously to (\ref{eq-6-28-1}), by introducing slack variables for each subproblem, (\ref{eq-6-33-1}) is equivalent to the following SDP: 
\begin{equation}\label{eq-6-34}
    \begin{aligned}
        & \{ \hat{\mathcal{Z}}_s^{(\kappa + 1)}, \cdot, \cdot \} \!:=\! \\
        & \argmin_{ \{ \hat{\mathcal{Z}}_s, \varphi_s, \tilde{\varphi}_s \} \in \mathbb{Z}_s \cap \mathbb{Z}_s^{\phi} \cap \tilde{\mathbb{Z}}_s^{\phi}  } 
        \hat{J}_s( \hat{\mathcal{Z}}_s ) + \frac{\rho}{2} \varphi_s +  \frac{\tilde{\rho}}{2} \tilde{\varphi}_s ~~~ \forall s \in \mathbb{P}
    \end{aligned}
\end{equation}
with the feasible region $\mathbb{Z}_s^{\phi}\!$ defined by (\ref{eq-6-28-2}) and $\tilde{\mathbb{Z}}_s^{\phi}$ defined by
\begin{equation}\label{eq-6-35}
    \left[
    \begin{array}{cc}
        \tilde{\varphi}_s & \vect \left( \hat{\mathcal{Z}}_s \!-\! ( \mathcal{Y}_s^{(\kappa)} \!\!-\! \frac{1}{\tilde{\rho}} \tilde{\mathcal{A}}_s^{(\kappa)} ) \right)^{\!\!T} \\
        * & I \\   
    \end{array}
     \right]  \!\!\succeq \!0.
\end{equation}

\begin{remark}
    If $\tilde{\mathcal{A}}_s^{(0)}$ is symmetric, $\hat{\mathcal{Z}}_s - ( \mathcal{Y}_s^{(\kappa)} - \frac{1}{\tilde{\rho}} \tilde{\mathcal{A}}_s^{(\kappa)} )$ is also symmetric (see Remark \ref{remark-8} for the reason). Then analogously to Remark \ref{remark-sdp-reduction}, PSD constraints (\ref{eq-6-35}) can be reduced to multiple small PSD constraints.    
\end{remark}

\noindent
\textit{Step 2) Update auxiliary variables} 

The update of auxiliary variables $\mathcal{Z}_{\text{a}}$ and $\mathcal{Y}_s$ is given as
\begin{equation}\label{eq-6-33-2}
    \begin{aligned}
        & \{ \mathcal{Z}_{\text{a}}^{(\kappa + 1)},...,  \mathcal{Y}_s^{(\kappa + 1)},... \} \!:=\! \\
        & \argmin_{  \mathcal{Z}_{\text{a}}  \in \mathbb{R}_\text{a}, \mathcal{Y}_s \in \mathbb{Y}_s   }  
            \sum_{s \in \mathbb{P}}    L_s(\hat{\mathcal{Z}}_s^{(\kappa + 1)}, \mathcal{Z}_{\text{a}}, \mathcal{Y}_s, \mathcal{A}_s^{(\kappa)}, \tilde{\mathcal{A}}_s^{(\kappa)} ) 
    \end{aligned}
\end{equation}
where computing for $\mathcal{Z}_a^{(\kappa + 1)}$ and each $\mathcal{Y}_s^{(\kappa + 1)}$ can be conducted individually by separating (\ref{eq-6-33-2}) into
\begin{subequations}\label{eq-6-36}
    \begin{align}
        & \mathcal{Z}_{\text{a}}^{(\kappa + 1)} \!\!=\!\! \argmin_{  \mathcal{Z}_{\text{a}}  \in \mathbb{R}_\text{a} }  \sum_{s \in \mathbb{P}}  \!  L_s(\hat{\mathcal{Z}}_s^{(\kappa + 1)}\!,\! \mathcal{Z}_{\text{a}}, \mathcal{Y}_s^{(\kappa)}, \mathcal{A}_s^{(\kappa)}\!, \tilde{\mathcal{A}}_s^{(\kappa)} ) \label{eq-6-36:1}\\
        & \mathcal{Y}_s^{(\kappa + 1)} \!\!=\!\! \argmin_{ \mathcal{Y}_s \in \mathbb{Y}_s }  L_s(\hat{\mathcal{Z}}_s^{(\kappa + 1)}\!,\! \mathcal{Z}_{\text{a}}^{(\kappa)}\!,\! \mathcal{Y}_s,\! \mathcal{A}_s^{(\kappa)}\!,\! \tilde{\mathcal{A}}_s^{(\kappa)} \!) ~ \forall s \!\!\in\!\! \mathbb{P}  \label{eq-6-36:2}
    \end{align}
\end{subequations}
Here (\ref{eq-6-36:1}) is the same as the update of auxiliary variables in consensus ADMM,  which can be formulated as the simpler form given by (\ref{eq-6-28:2}). By Proposition \ref{pro-2-1-added} (see Appendix\ref{appendix-ppa}), matrices in $\mathcal{Y}_s$ can be updated in parallel. More importantly, updating of each matrix, i.e., $\!\bm{Y}_{\mathcal{C}_{ej}}$, is essentially a low rank approximation problem. This problem is non-convex due to rank constraints but an optimal solution can be given by the Eckart-Young-Mirsky Theorem \cite{4-508}. Accordingly, updates of $\mathcal{Y}_s$ can be conducted exactly and analytically as
\begin{equation}\label{eq-6-39}
    \mathcal{Y}_s^{(\kappa + 1)} = \left[ \sigma_{\mathcal{C}_{ej}}^1 u_{\mathcal{C}_{ej}}^1 v_{\mathcal{C}_{ej}}^{1~ T}  \right]^T  ~\text{with}~ \mathcal{C}_{ej} \in \mathcal{K}_s ~~~\forall s \in \mathbb{P}
\end{equation}

\begin{remark}\label{remark-8}
    For all $\kappa \geq 0$,  $\tilde{\mathcal{Z}}_s^{(\kappa + 1)}$ is symmetric, and the same for $\mathcal{Y}_s^{(\kappa + 1)}$ according to (\ref{eq-6-39}). Then by (\ref{eq-6-33-3}), as long as $\tilde{\mathcal{A}}_s^{(0)}$ is symmetric, $\tilde{\mathcal{A}}_s^{(\kappa)}$ is symmetric and thus the same for matrix $ \hat{\bm{Z}}_{\mathcal{C}_{ej}}^{(\kappa + 1)}  \!+\! \frac{1}{\tilde{\rho}}  \hat{\bm{\Lambda}}_{\mathcal{C}_{ej}}^{(\kappa)}$ for all $\kappa \geq 0$. Therefore, the singular value decomposition (SVD) of matrix $\hat{\bm{Z}}_{\mathcal{C}_{ej}}^{(\kappa + 1)}  + \frac{1}{\tilde{\rho}}  \hat{\bm{\Lambda}}_{\mathcal{C}_{ej}}^{(\kappa)}$ is degenerated into an eigenvalue decomposition if $\tilde{\mathcal{A}}_s^{(0)}$ is symmetric.
\end{remark}

% \begin{theorem}[Eckart-Young-Mirsky Theorem \cite{4-508}]
%     If the matrix $A \in \mathbb{R}^{m \times n}$ admits the singular value decomposition $A = U \Sigma V^T$ where $U = [u_1, u_2, ..., u_m] \in \mathbb{R}^{m \times m}$ and  $V = [v_1, v_2, ..., v_n]\in \mathbb{R}^{n \times n}$ are orthogonal, and $\Sigma = \diag(\sigma_1, \sigma_2,...,\sigma_{\min\{m,n\}})$ with $\sigma_1 \geq \sigma_2 \geq ... \geq \sigma_{\min\{m,n\}}$. Then an optimal solution to problem $\min_{X \in \mathbb{R}^{m \times n},\rank(X) \leq K} \Vert X - A \Vert_F $ is given by $X^* = \sum_{i=1}^{K} \sigma_i u_i v_i^T $.
% \end{theorem}

\noindent
\textit{Step 3) Update dual variables}

The update of auxiliary variables $\mathcal{A}_s$ is the same as (\ref{eq-6-28:3}), and that of $\tilde{\mathcal{A}}_s$ is given by
\begin{equation}\label{eq-6-33-3}
    \tilde{\mathcal{A}}_s^{(\kappa + 1)} \!:=\! \tilde{\mathcal{A}}_s^{(\kappa)} \!\!+\!\! \tilde{\rho} ( \hat{\mathcal{Z}}_s^{(\kappa + 1)} -  \mathcal{Y}_s^{(\kappa + 1)}  )~~~  \forall s \!\!\in\!\! \mathbb{P} 
\end{equation}

%\REQUIRE  The forecasted system operating condition( $\mathcal{N}$, $\mathcal{B}$, $\mathcal{N}_g$, $\mathcal{N}_s$, $\mathcal{N}_v$, $\mathcal{N}_l$, $\mathcal{N}_o$, $\mathcal{N}_{v_o}$, $\mathcal{N}_{v_d}$, $\mathcal{N}_{v_m}\!$, $\!\mathcal{N}_{v_{dm}}$, $p_i$, $V_i$, $\theta_{t_0}$, $\omega_{t_0}$ and $B$), $\mathcal{D}$, $\rho_k$, $\underline{\omega}^k$, $\overline{\omega}^k$, $\overline{\delta}$, $\underline{p}_g$, $\overline{p}_g$, $\underline{m}$, $\overline{m}$, $\underline{d}$, $\overline{d}$; $W_1$ to $W_5$, $t_0$, $t_f$, $n_t^k$, $n_c$, (NLP finished); $N_S$, $\Xi_s$, $\rho$, $\tilde{\rho}$
 
\begin{algorithm}
    \caption{Feasibility-embedded distributed approach}
    \label{alg_1} 
    \begin{algorithmic}[1]
            \REQUIRE $N_S$, $\Xi_s$, $\rho$, $\tilde{\rho}$, $\epsilon^{\text{abs}}$ and $\epsilon^{\text{rel}}$
            \ENSURE $M, D$
            \STATE Initialize $\mathcal{Z}_a^{(0)}$, $\mathcal{Y}_s^{(0)}$, $\mathcal{A}_s^{(0)}$, $\tilde{\mathcal{A}}_s^{(0)}$ and ${\kappa} \leftarrow -1$
            \REPEAT
                \STATE ${\kappa} \leftarrow \kappa + 1$

                \LINEFOR{$s \leftarrow 1$ to $N_S$} {$\hat{\mathcal{Z}}_s^{(\kappa + 1)}$ $\leftarrow$ Eq. (\ref{eq-6-34})}%

                \STATE $\mathcal{Z}_{\text{a}}^{(\kappa + 1)}$ $\leftarrow$ Eq. (\ref{eq-6-28:2})
                
                \FOR{$s \leftarrow 1$ to $N_S$}                    
                        \FOR{\textbf{each} $\mathcal{C}_{ej}$ in $\mathcal{K}_s $}
                        \STATE $\{\sigma_{\mathcal{C}_{ej}}^1, u_{\mathcal{C}_{ej}}^1, v_{\mathcal{C}_{ej}}^{1~ T}\}$ $\leftarrow$ SVD for $\hat{\bm{Z}}_{\mathcal{C}_{ej}}^{(\kappa + 1)}  + \frac{1}{\tilde{\rho}}  \hat{\bm{\Lambda}}_{\mathcal{C}_{ej}}^{(\kappa)}$
                        \ENDFOR
                        \STATE $\mathcal{Y}_s^{(\kappa + 1)}$ $\leftarrow$ Eq. (\ref{eq-6-39})
                \ENDFOR

                \LINEFOR{$s \leftarrow 1$ to $N_S$} {$\{\mathcal{A}_s^{(\kappa + 1)}, \tilde{\mathcal{A}}_s^{(\kappa + 1)}\} \leftarrow$Eq. (\ref{eq-6-28:3}, \ref{eq-6-33-3})}

                \STATE Compute $r^{(\kappa + 1)}$, $s^{(\kappa+1)}$, $\epsilon^{\text{pri}(\kappa + 1)}$ and $\epsilon^{\text{dual}(\kappa + 1)}$.

            \UNTIL $\Vert r^{(\kappa + 1)} \Vert_2 < \epsilon^{\text{pri}(\kappa+1)} \wedge \Vert s^{(\kappa + 1)} \Vert_2 < \epsilon^{\text{dual}(\kappa+1)}$
            \STATE $m_j \!\!\leftarrow\!\! \bm{Z}_{j, (1,2)}^{\text{md}}$, $d_j \!\!\leftarrow\!\! \bm{Z}_{j,(1,3)}^{\text{md}}$, $\forall j \!\!\in\! \mathcal{N}_g$
    \end{algorithmic}
\end{algorithm}

Finally, Algorithm \ref{alg_1} shows the pseudocode of the proposed feasibility-embedded distributed approach. Here computations in line 4, accounting for almost entire computational efforts, can be conducted in parallel across at most $N_S$ processors.

%Computation in line 14 is given in Appendix. 

\section{Case Study}\label{section-5}

The proposed numerical method for DID is tested on five systems, including the simplified 14-generator Australian (AU14Gen) \cite{4-118}, IEEE 14-bus, IEEE 39-bus, IEEE 118-bus and ACTIVSg200 systems \cite{4-517}. Six normal steady-state operating conditions of the AU14Gen system, named case 1 to case 6 following Table 1 in \cite{4-118}, are also used to demonstrate the necessity of DID. IPOPT interfaced by Pyomo, and MOSEK interfaced by CVXPY, are employed to sovle NLPs and SDPs, respectively. All computations are carried out on a Linux 64-Bit server with 2 Intel(R) Xeon(R) E5-2640 v4 @ 2.40GHz CPUs (a total of 40 processors provided) and 125GB RAM. Distributed computing across multiple processors for line 5 of Algorithm \ref{alg_1} is realized using Ray \cite{4-523}.

\subsection{Parameter Setting}

Two generator settings are considered to simulate different operating modes or composition of generators. For the AU14Gen system, all generators are modelled as inverters in $\mathcal{N}_{v_{dm}}$, where both virtual inertia and damping of all inverters are dispatchable. For other test power systems, generators are set in $\mathcal{N}_{v_o}$ (or $\mathcal{N}_g$), $\mathcal{N}_{v_d}$ (low inertia), $\mathcal{N}_{v_d}$ (high inertia), $\mathcal{N}_{v_m}\!$ and $\!\mathcal{N}_{v_{dm}}$ circularly. Parameters of each set of generators are given in Table \ref{tb_appendix_4_1} in Appendix\ref{appendix_5}. Transient reactances of synchronous generators are ignored for simplicity. The base MVA is 100 MVA. For each load, $d_{li} \!=\! 0.01~\text{p.u.} \cdot \text{s} / \text{rad}$. For each test system, $\mathcal{D}$ contains four disturbances with parameters given in Table \ref{tb_appendix_4_2} of Appendix\ref{appendix_5}. Matrices $W_1$ to $W_5$ are all set to identity matrices with proper dimension. 
For the setting of time horizon, $t_0 \!=\! 0~\text{s}$ and $t_f \!=\! 30~\text{s}$. Frequency bounds $\overline{\omega}^k(t)$ and $\underline{\omega}^k(t)$ are give in Table \ref{tb_appendix_4_3} in Appendix\ref{appendix_5}, referring to draft NEM mainland frequency operating standards of interconnected systems \cite{4-243}; and $\overline{\delta} = 3\pi/4$.

In the NLP formulation of DID, 3rd-order Radua collocation and $n_t^k \!=\! 20$ are employed. In Approximation \ref{approx-2}, $\theta_b \!=\! 0.580001$ to minimize the approximation error according to Appendix\ref{appendix_3}, where $\epsilon\!\!=\!\! 2.2155 \!\times\! 10^{-4}$. In the fill-reducing Cholesky factorization, $\beta_{cf} \!=\! 100$ can guarantee positive definiteness of $A_{adj} \!+\! \beta_{cf} I$ for all systems. In the feasibility-embedded distributed approach, $N_S\!=\!40$, $\Xi_s \!=\! \{(k,2s \!-\! 1), (k,2s)\}$ with $k \!\in\! \mathcal{D}$ and $s \in \mathbb{P}$, $\epsilon^{\text{abs}} \!=\! 10^{-5}$ and $\epsilon^{\text{rel}} \!=\! 10^{-3}$ \cite{4-461}.

%$\rho$ and $\tilde{\rho}$ are set to $1$ and $0.1$, respectively.

\subsection{Numerical Results}

% add the process of find clique tree and maximal cliques

The proposed feasibility-embedded distributed approach (FEDA) is used to solve the DID problem, for all test systems. Fig. \ref{fig-cs-2} and Fig. \ref{fig-cs-3} show the progress of the primal and dual residual norms by iteration, for six cases of the AU14Gen power system and other four test systems, respectively. The dashed lines show the feasibility tolerances $\epsilon_{\text{pri}}$ and $\epsilon_{\text{dual}}$. The vertical dotted lines show when the feasibility tolerance is satisfied, and the rightmost vertical dotted line shows when the stopping criterion of Algorithm \ref{alg_1}, i.e., line 14, has been satisfied. We run the FEDA for 100 iterations to show the continued progress while the stopping criterion could be satisfied beforehand. According to Fig. \ref{fig-cs-2} and Fig. \ref{fig-cs-3}, it can be concluded that stopping criterion of the FEDA, with certain values of penalty parameters $\rho$ and $\tilde{\rho}$, is satisfied within 100 iterations for all test systems. For the AU14Gen power system under different operating conditions, the number of iterations for convergence is slight different but all within 65 iterations. The increase of the size of power systems only causes slow growth in the number of iterations for convergence.

\begin{figure}[t]
    \includegraphics{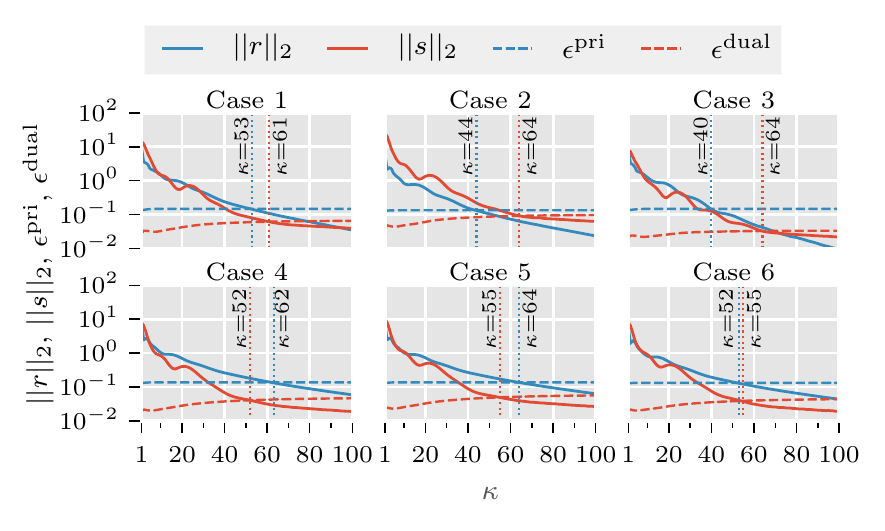}
    \caption{  
    Norms of primal residual and dual residual versus iteration, for the AU14Gen power system under operating condition case 1 to case 6. Penalty parameters $(\rho, \tilde{\rho})$ are $(1.0, 0.8)$, $(1.5, 1.0)$, $(0.8, 0.4)$, $(0.5, 0.5)$, $(0.8, 0.8)$ and $(0.5, 0.5)$, for case 1 to case 6, respectively. \hlr{Labels for the minor ticks (the shorter ones) of the abscissa are not shown while their values equal to the average of labels of the two adjacent major ticks.}}
    \label{fig-cs-2}
\end{figure}
\begin{figure}[t]
    \includegraphics{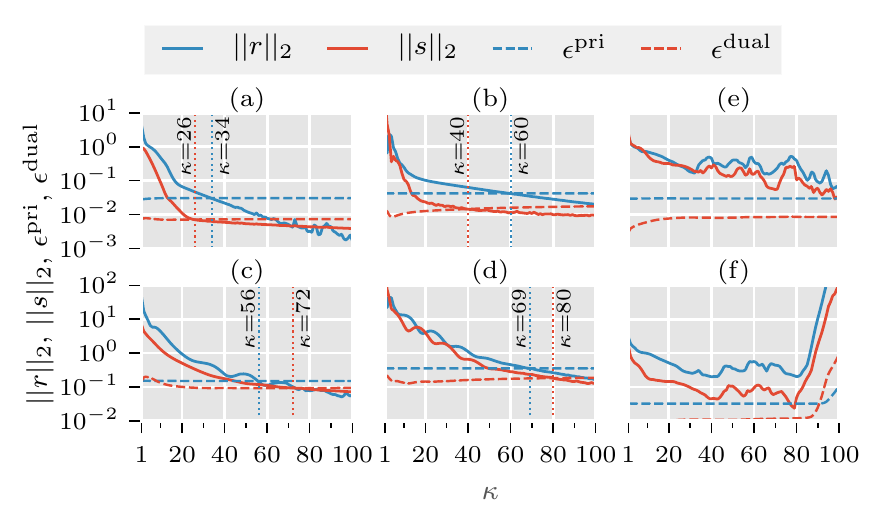}
    \caption{  
    Norms of primal residual and dual residual versus iteration, for the IEEE 14-bus (a, e, f), IEEE 39-bus (b), IEEE 118-bus (c) and ACTIVSg200 (d) power system. Penalty parameters $(\rho, \tilde{\rho})$ are $(1.0, 0.2)$, $(0.5, 1.5)$, $(2.0, 1.0)$, $(1.5, 2.0)$, $(1.0, 0.5)$ and $(0.2, 1.0)$, for (a) to (f). \hlr{Labels for the minor ticks of the abscissa are analogous to that in Fig. \ref{fig-cs-2}.}}
    \label{fig-cs-3}
\end{figure}

It can be seen that the convergence rate of $\Vert r \Vert_2$ and $\Vert s \Vert_2$ generally becomes increasingly slow as iteration progresses, which means that the FEDA can converge to modest accuracy within an acceptable number of iterations but be very slow to converge to high accuracy. This is determined by inherent convergence characteristics of ADMM \cite{4-461}. Nonetheless, modest accuracy is sufficient for the practical application of DID.

It should be noted that unlike solving convex problems by ADMM where convergence can be guaranteed under mild conditions and is immune to values of penalty parameters in the augmented Lagrangian, the FEDA is not necessarily convergent. It is found that improper values of penalty parameters can lead the FEDA to be divergent, which, however, only affect convergence time when using ADMM to solve convex problems \cite{4-461}. Taking the IEEE 14-bus system for example, with $(\rho, \tilde{\rho}) \!=\! (1.0, 0.5)$, the FEDA starts to drastically fluctuate after about 30 iterations as shown in Fig. \ref{fig-cs-3}(e), and with $(\rho, \tilde{\rho}) \!=\! (0.2, 1.0)$, the FEDA is divergent as shown in Fig. \ref{fig-cs-3}(f). For different test systems, we took different values of $(\rho, \tilde{\rho})$ to ensure convergence by trial and error. However, it seems to be far from easy to draw a general conclusion about tunning of penalty parameters. We leave the scheme of penalty parameters' tunning or adaptive adjustment for future work. 

\begin{table}[!t]
	\caption{Comparison of optimization results of NLP, SDP and the FEDA. }
    \label{tb-cs-1}
    {\footnotesize{
	\begin{tabular*}{\hsize}{@{}@{\extracolsep{\fill}}llllll@{}}
	\toprule
    \multirow{2}{*}{Test systems} & \multicolumn{3}{l}{Optimal objective values} & \multicolumn{2}{c}{\textcolor[rgb]{0,0,0.5}{$\rank(\hat{\mathcal{Z}})^*$} } \\
    \cline{2-4}  \cline{5-6}
    & NLP & SDP & FEDA & SDP & \textcolor[rgb]{0,0,0.5}{FEDA} \\
    \midrule
    Case 1           & 187.95   & 93.67  & 127.44  & 21 & \textcolor[rgb]{0,0,0.5}{1} \\
    Case 2           & 265.61   & 120.66 & 191.14  & 15 & \textcolor[rgb]{0,0,0.5}{1} \\
    Case 3           & 106.74   & 42.85  & 74.95  & 14  & \textcolor[rgb]{0,0,0.5}{1} \\
    Case 4           & 61.547   & 39.67  & 46.10  & 20  & \textcolor[rgb]{0,0,0.5}{1} \\
    Case 5           & 108.05   & 59.14  & 83.15  & 21  & \textcolor[rgb]{0,0,0.5}{1} \\
    Case 6           & 71.123   & 37.86  & 57.67  & 20  & \textcolor[rgb]{0,0,0.5}{1} \\
    IEEE 14-bus      & 116.86   & 32.55  & 54.98  & 12  & \textcolor[rgb]{0,0,0.5}{1} \\
    IEEE 39-bus      & 269.99   & 152.25 & 183.09  & 17 & \textcolor[rgb]{0,0,0.5}{3} \\
    IEEE 118-bus     & 100.21   & 63.93  & 88.72  & 12  & \textcolor[rgb]{0,0,0.5}{1} \\
    ACTIVSg200       & 115.47   & 70.40  & 81.98  & 7   & \textcolor[rgb]{0,0,0.5}{1} \\
    \bottomrule
    \end{tabular*}
    }}
    \begin{tablenotes}\footnotesize
        \item[*] $^*$ Threshold is set to $10^{-5}$ below which eigenvalues are considered zero.
    \end{tablenotes}
\end{table}

Furthermore, we compare the results of DID obtained by NLP (solving the NLP formulation (P2)), SDP (solving the decomposed SDP relaxation (P9) directly) and the FEDA. The optimal objective values and rank of solutions for them are listed in Table \ref{tb-cs-1}. Note that in computing $\rank(\hat{\mathcal{Z}})$ the threshold is set to $10^{-5}$ below which eigenvalues are considered zeros due to the fact that only modestly accurate solutions are expected while applying the FEDA. Regarding the optimal objective values in Table \ref{tb-cs-1}, we can see that for all test systems, the optimal objective value obtained by NLP is greatly larger than that obtained by SDP, and the FEDA produces a much lower optimal objective values than NLP. \textcolor[rgb]{0,0,0.5}{For rank of solutions in Table \ref{tb-cs-1}, $\rank(\hat{\mathcal{Z}})$ of the solution obtained by SDP is large than 1 for all test systems while for the solution obtained by the FEDA, $\rank(\hat{\mathcal{Z}}) \!=\! 1$ for all test systems except the IEEE 39-bus power system. But by increasing the threshold in computing $\rank(\hat{\mathcal{Z}})$ to $2 \!\times\! 10^{-5}$ or taking the solution after the 62th iteration, we still have $\rank(\hat{\mathcal{Z}}) \!=\! 1$ for the IEEE 39-bus system. Therefore, it can be concluded that for all test systems, the SDP relaxation is inexact and thus solving the decomposed SDP relaxation (P9) directly can only result in infeasible solutions to (P7). More importantly, the proposed FEDA can produce solutions which are not only with much smaller objective values than that found by NLP but also feasible to the original problem (P7) under modest tolerances. Ignoring the approximation errors between (P2) and (P7), the optimal objective value obtained by SDP gives a lower bound of objective function $\hat{J}$ in (P2) and solutions obtained by NLP and the FEDA are both a local optimum of (P2).} Clearly, the FEDA achieves a much smaller optimality gap for solutions of (P2) than NLP does for all test systems.

\hlr{
The effectiveness of the FEDA is also demonstrated by the time-domain results given by Fig. \ref{case14_1-1} to Fig. \ref{case14_4-1}. For clarity, we call the systems with the DID results obtained by NLP and the FEDA, the NLP system and the FEDA system, respectively. Here we focus on the IEEE 14-bus system for the sake of observability, and compare time-domain curves of the NLP system and FEDA system. By Fig. \ref{case14_1-1}, we can find that under the power-step disturbance, the curves of phase angle differences of branches of the NLP system and FEDA system are close, while the DID system outperforms the NLP system regarding the frequency nadir, steady-state frequency, frequency oscillation and the maximal rate of change of frequency (RoCoF). Regarding control efforts, power output changes of most generators and inverters in the FEDA system are overall larger than that in the NLP system within 0.5 s after the disturbance occurs, while in the steady state, the opposite is the case. Under the power-ramp disturbance, the DID system outperforms the NLP system regarding the steady-state frequency and the maximal RoCoF, as shown in Fig. \ref{case14_2-1}. By Fig. \ref{case14_3-1}, it is observed that under the power-fluctuation disturbance, fluctuations in $\dot{\omega}$ and especially $\omega$ of the FEDA system are smaller than that of the NLP system, while fluctuations in $\Delta p$ of the two system are close. Under the three-phase short circuit disturbance, out-performance of the FEDA system is more significant, as shown in Fig. \ref{case14_4-1}. Except for the out-performance regarding the frequency nadir, frequency oscillation and the maximal RoCoF, the FEDA system also has overall smaller oscillations in phase angle differences of branches and power output changes of generators and inverters.
}

\begin{figure}[t]
    \centering
    \includegraphics[scale=0.9]{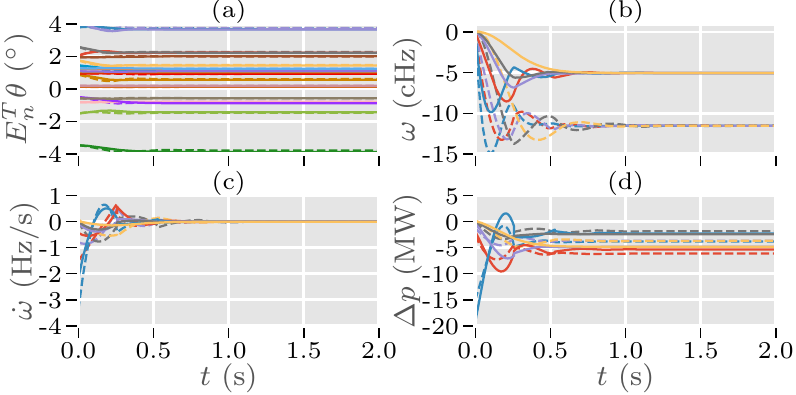}
    \caption{\hlr{Time-domain results of the IEEE 14-bus system under the power-step disturbance, where (a), (b), (c) and (d) are curves of the phase angle difference of branches, angular frequency of generator or inverter buses, rate of change of frequency of generator or inverter buses, and change of power output of generators or inverters, respectively. The solid and dashed lines are curves for the FEDA system and NLP system, respectively. Different colored lines correspond to different branches in (a), and correspond to different generators or inverters in (b), (c) and (d). In (d), $\Delta p = M \dot{\omega} + D \omega$.}}
    \label{case14_1-1}
\end{figure}

\begin{figure}[t]
    \centering
    \includegraphics[scale=0.9]{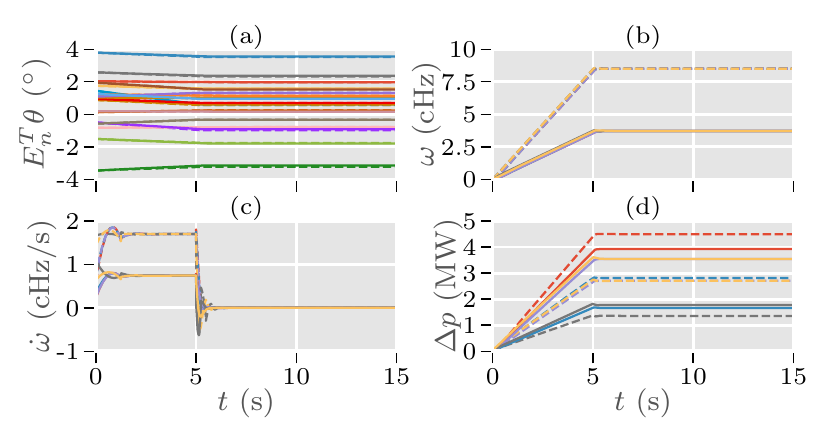}
    \caption{\hlr{Analgous to Fig. \ref{case14_1-1} but under the power-ramp disturbance.}}
    \label{case14_2-1}
\end{figure}

\begin{figure}[t]
    \centering
    \includegraphics[scale=0.9]{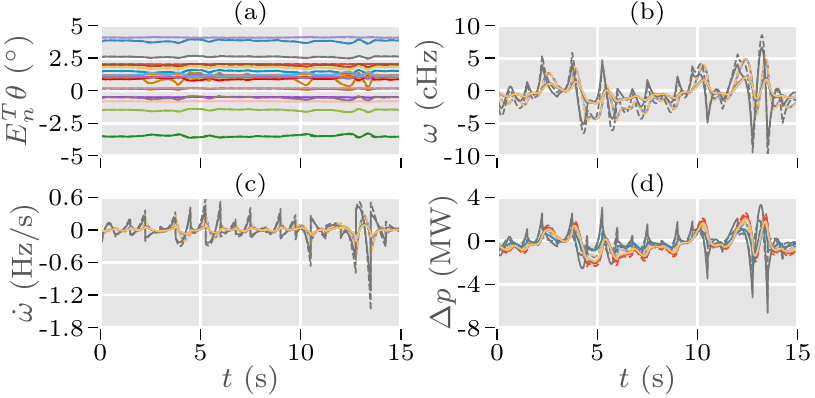}
    \caption{\hlr{Analgous to Fig. \ref{case14_1-1} but under the power-fluctuation disturbance.}}
    \label{case14_3-1}
\end{figure}

\begin{figure}[t]
    \centering
    \includegraphics[scale=0.9]{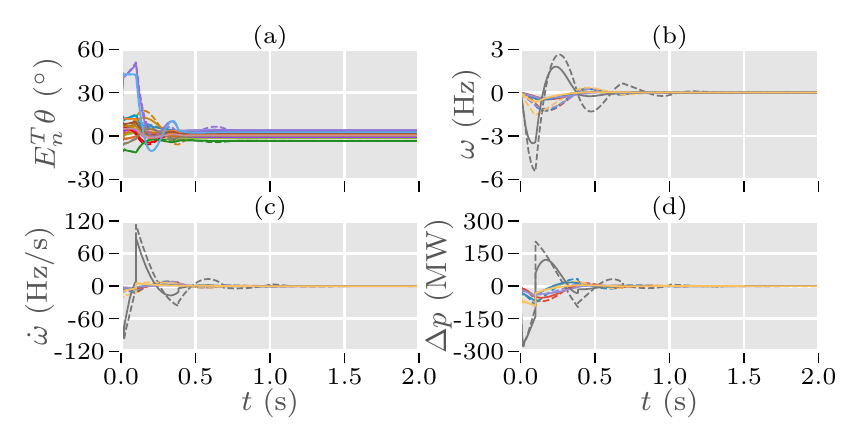}
    \caption{\hlr{Analgous to Fig. \ref{case14_1-1} but under the three-phase short circuit disturbance.}}
    \label{case14_4-1}
\end{figure}

\begin{figure}[t]
    \includegraphics[scale=0.98]{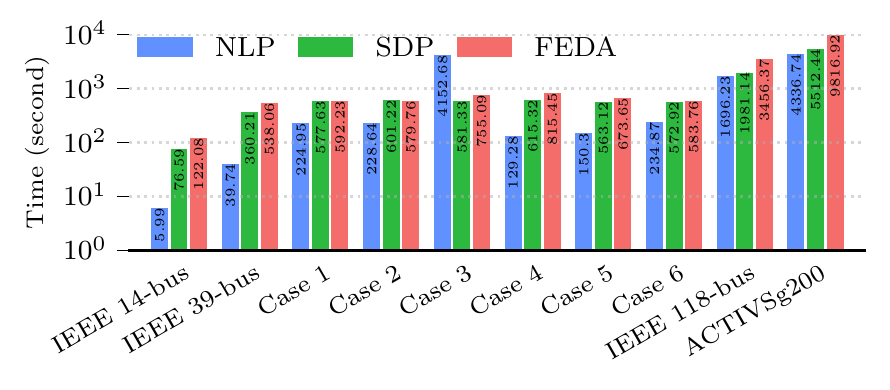}
    \caption{Comparison of computation time of NLP, SDP and the FEDA.}
    \label{fig-cs-4}
\end{figure}

\begin{figure}[t]
    \includegraphics[scale=0.95]{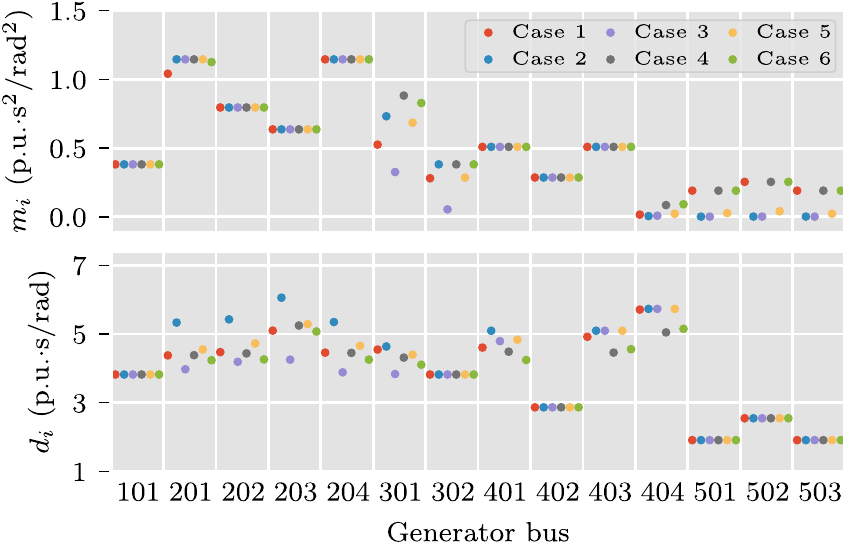}
    \caption{\hlr{Dispatch results of virtual inertia and damping for the AU14Gen system under operating condition case 1 to case 6, obtained by the FEDA.}}
    \label{fig-cs-1}
\end{figure}

Fig. \ref{fig-cs-4} compares computation time of NLP, SDP and the FEDA spent on solving DID of each test system. The use of open-source modeling languages, i.e., Pyomo and CVXPY, causes significant time spent on constructing and passing models, which, however, can be reduced to a negligible amount by employing C++ interface of IPOPT and Fusion interface of MOSEK. Thus computation time of NLP and SDP only includes the time spent by optimizers, and time spent on line 4 of Algorithm \ref{alg_1} equals to the maximal time spent by MOSEK in each processor. In Fig. \ref{fig-cs-4}, we can see that NLP has a computation time advantage among the three approaches, being undermined as the size of power systems increases. As the cost of achieving solutions with a smaller optimality gap, the FEDA is inevitably with the maximum computation time for most cases, which, however, is acceptable profiting from distributed parallel computing. For the six cases of AU14Gen system, computation time of the FEDA is even very close or less than that of SDP.

Fig. \ref{fig-cs-1} shows dispatch results of virtual inertia and damping for the AU14Gen system under different operating conditions, obtained by the FEDA. We can see that following variation in operating conditions, most generators need to significantly adjust their virtual inertia, damping or both of them to optimize the system performances and control efforts. This, to some extend, demonstrates the necessity of DID for operation of future power grids with high heterogeneity in operating conditions, to ensure a optimal tradeoff between synchronism performances, frequency performances and control efforts.

\section{Conclusion}\label{section-6}

This paper numerically addresses the DID problem for future inverter-dominant transmission networks. By the Radua collocation method, the DID problem is first formulated as a NLP with flexibility to handling time-varying performance constraints and various types of disturbances. Next, the highly non-convex NLP is relaxed into a convex SDP for which sparsity is exploited to improve computational efficiency. Finally, a feasibility-embedded distributed solution approach is proposed under the framework of ADMM. Numerical experiments on five test systems demonstrate that by tunning penalty parameters, the proposed solution approach can converge to modest accuracy within several tens of iterations as well as acceptable computation time benefiting from distributed parallel computing. The SDP relaxation of the NLP of DID is inexact while the feasibility-embedded distributed approach can produce solutions being not only feasible to the original problem but also with much smaller optimality gaps than that achieved by the local solution approach. Variations in dispatch results under different operating conditions for the AU14Gen system demonstrates the necessity of DID for power grids with increasingly high heterogeneity in operating conditions.

For the future direction, firstly, influence of penalty parameters on convergence of the feasibility-embedded distributed approach will be further investigated to develop tuning or adaptive adjustment schemes; secondly, impacts of fidelity of system dynamic models on results of DID will be evaluated for determining appropriate model fidelity that balances the computational complexity and accuracy for the DID problem, and model reduction for parts of system more remote from the fault can potentially deal with the case involving high-fidelity models; and thirdly, distributed DID independent of central control centers is worth pursuing, where learning-based approaches can promisingly tackle complications caused by possible divergence of distributed algorithms for non-convex problems and potential real-time execution of DID.

% It is far from easy to prove convergence of the feasibility-embedded distributed approach theoretically, 

% or


\section*{Acknowledgment}
The authors would like to thank Dr. Xiao Han, from Microsoft Research Asia, for the support of Linux servers.

\ifCLASSOPTIONcaptionsoff
  \newpage
\fi

\bibliographystyle{IEEEtran}
\bibliography{bare_jrnl_long}

\newpage

\appendices
\section*{Appendix}


\subsection{Numerical Analysis of Approximation 2}\label{appendix_3}

The essence of Approximation \ref{approx-2} is using quadratic function $\varsigma_1$, linear function $\varsigma_2$ and quadratic function $\varsigma_3$ to approximate $\sin \theta$ for $\theta \!\!\in\!\! [ - \theta_c, -\theta_b]$, $\theta \!\!\in\!\! [ - \theta_b, \theta_b]$ and $\theta \!\!\in\!\! [\theta_b, \theta_c]$, respectively. Clearly, $\theta_b$ can observably impact approximation errors and should be selected carefully. Define approximation error function $\epsilon\!:\! \theta_b \!\!\to\!\! \int_{-\theta_c}^{\theta_c} (\sin \theta \!-\! {\beta}^T \varsigma)^2 \text{d} \theta$ and function $\theta_b^*:\! \theta_c \!\! \to$ $\{\argmin_{\theta_b} \epsilon, \text{s.t.} \theta_b \!\!\in\!\! [0, \frac{\pi}{2}]\}$. Graphs of $\theta_b$-$\theta_c$-$\epsilon$, $\theta_c$-$\theta_b^*$ and $\theta_c$-$\min_{\theta_b}\! \epsilon\!$ are shown in Fig. \ref{fig-approx-1}. Numerically, we can find that $\forall \theta_c \!\in\! [\frac{\pi}{2}, \pi]$, $\epsilon$ is a convex function in domain $[0, \frac{\pi}{2}]$. Thus, given $\theta_c \!\!\in\!\! [\frac{\pi}{2}, \pi]$, by solving $\{\argmin_{\theta_b} \epsilon, \text{s.t.} \theta_b \!\!\in\!\! [0, \frac{\pi}{2}]\}$, the unique optimal value of $\theta_b \!\!=\!\! \theta_b^*$ to minimize the approximation error can be obtained. Additionally, we can also find that for  $\theta_c \!\!\in\!\! [\frac{\pi}{2}, \frac{9}{10}\pi]$, Approximation \ref{approx-2} can be with very high accuracy. 

\begin{figure}[h]
    \centering
    \subfigure[Function surface of $\theta_c$-$\theta_b$-$\epsilon$] {\label{fig-approx-1:a} \includegraphics[width=1.57in]{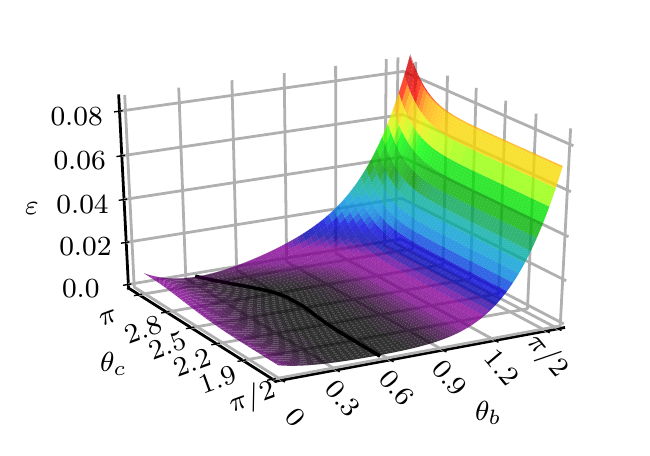}}
    \subfigure[Curves of $\theta_c$-$\theta_b^*$ and $\theta_c$-$\min_{\theta_b}\!\! \epsilon$] {\label{fig-approx-1:b} \includegraphics[width=1.75in]{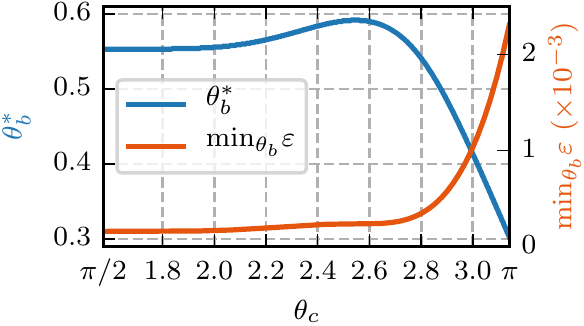}}
    \caption{Numerical analysis of Approximation 2.}
    \label{fig-approx-1}
\end{figure}

\subsection{\hlr{Propositions and Corollaries}}\label{appendix-ppa}

\begin{proposition}\label{prop-chordality-1}
    For a given graph $\mathcal{G}(\mathcal{V}, \mathcal{E})$, define a graph $\bar{\mathcal{G}}(\bar{\mathcal{V}}, \bar{\mathcal{E}})$ where each node in $\bar{\mathcal{V}}$ represents a set of nodes in $\mathcal{V}$; $\cup_{\bar{v} \in \bar{\mathcal{V}}} \bar{v} = \mathcal{V}$; $\forall \bar{v}, \bar{v}' \in  \bar{\mathcal{V}}, \bar{v} \neq \bar{v}'$, $\bar{v} \cap \bar{v}' = \emptyset$; and $(\bar{i}, \bar{j}) \in \bar{\mathcal{E}}$ satisfies $\exists (i,j) \in \bar{i} \times \bar{j} $, $(i,j) \in \mathcal{E}$. Let $\mathcal{K} = \{\mathcal{C}_1, ... , \mathcal{C}_{n_{mc}} \}$ and $\bar{\mathcal{K}} = \{\bar{\mathcal{C}}_1, ... , \bar{\mathcal{C}}_{\bar{n}_{mc}} \}$ be sets of all maximal cliques of $\mathcal{G}$ and $\bar{\mathcal{G}}$, respectively. Then the following statements hold:

(\romannumeral1) $\mathcal{G}(\mathcal{V}, \mathcal{E})$ is chordal iff $\bar{\mathcal{G}}(\bar{\mathcal{V}}, \bar{\mathcal{E}})$ is chordal and $\forall i \in \{1,...,\bar{n}_{mc}\}$, the induced subgraph of $\mathcal{G}(\mathcal{V}, \mathcal{E})$ by the node set $\bar{\mathcal{C}}_i$ is chordal; and 

(\romannumeral2) denote by $\mathcal{K}_i$ the set of all maximal cliques of the induced subgraph of $\mathcal{G}(\mathcal{V}, \mathcal{E})$ by the node set $\bar{\mathcal{C}}_i$, then $\mathcal{K} = \cup_{i\in \{1,...,\bar{n}_{mc}\}} \mathcal{K}_i $.
\end{proposition}

\begin{proof}
    The proof is trivial with basic properties of graphs and the definition of maximal cliques.
\end{proof}

\begin{proposition}\label{proposition-9}
    Denote by $\mathcal{G}_n (\mathcal{N}, \mathcal{B})$ the underlying graph of the power network, then we have

    (\romannumeral1) graph $\mathcal{G}$ is chordal if $\mathcal{G}_n$ is chordal;

    (\romannumeral2) assume that $\mathcal{G}_n$ is chordal with its set of maximal cliques denoted by $\mathcal{K}_{n} \!\!=\!\! \{\mathcal{C}_{nj}\}$ since we can always find a chordal extension for $\mathcal{G}_n$.
    Maximal cliques of $\mathcal{G}$ are given by
    {\small{
    $
            \mathcal{K} = \bigcup_{k \in \mathcal{D}, i\in \mathbbm{T}^k} \mathcal{K}_i^k ~\text{where}~ \mathcal{K}_i^k = \mathcal{K}_{e1}^{ki} \cup \mathcal{K}_{e3}^{ki} \cup \mathcal{K}_{e3}^{ki} ~\text{with}~
    $
    \begin{subequations}\label{eq-2-1-ap-1}
        \begin{align}
            &\!\!\!\mathcal{K}_{e1}^{ki} \!\!=\!\! \{\!\mathcal{C}_{ej}| \mathcal{C}_{ej} \!\!=\!\! \mathcal{I}(\!M\!\mathbbm{1},\!j)\!\cup\! \mathcal{I}(\!D\mathbbm{1},\!j) \!\cup\! \mathcal{I}(\bm{\omega}_i^k\!,j) \!\cup\! \mathcal{I}(-1)\!, j \!\!\in\! \mathcal{N}_g\! \} \label{eq-2-1-ap-1:1}\\
            &\!\!\!\mathcal{K}_{e2}^{ki} = \{\mathcal{C}_{ej}| \mathcal{C}_{ej} = \mathcal{I}(\bm{l}_{\text{m}i}^k,j)\cup\mathcal{I}(\bm{l}_{\text{d}i}^k,j) \cup \mathcal{I}(-1), j \in \mathcal{N}_g \} \label{eq-2-1-ap-1:2}\\
            &\!\!\! \mathcal{K}_{e3}^{ki} = \{\mathcal{C}_{ej}| \mathcal{C}_{ej} = \mathcal{I}(\bm{\theta}_i^k, \mathcal{C}_{nj}) \cup \mathcal{I}(-1), \mathcal{C}_{nj} \in \mathcal{K}_n \}  \label{eq-2-1-ap-1:3}
        \end{align}
    \end{subequations}
    }}
\end{proposition}
\begin{proof}
    Fig. \ref{fig-sparity-1:a} shows the aggregate sparsity pattern of matrix $\!\bm{Z}\!$ at the block level with $\bm{Z}$ broken into blocks corresponding to each pair of disturbances and time elements. We can find that the sparsity graph is chordal and its set of maximal cliques is given by
    {\small
\begin{equation}\label{eq-sparsity-1}
        \!\!\! \mathcal{K}_{b} \!\!=\!\! \left\{\!\mathcal{C}_i^k| k \!\!\in\!\! \mathcal{D}, i \!\!\in\!\! \mathbbm{T}^k\!,\! \mathcal{C}_i^k \!=\!\! \{\mathcal{I}(M\mathbbm{1},\!\cdot),\! \mathcal{I}(D\mathbbm{1},\cdot),\! \mathcal{I}(k,i),\! \mathcal{I}(-1) \} \!\right\}  \!\!\!
\end{equation}
    }
Note that dashed edges in Fig. \ref{fig-sparity-1:a} are extra added to reduced the number of equality constraints for overlapping entries. Otherwise, $\forall k \!\!\in\! \mathcal{D}, i \!\!\in\! \mathbbm{T}^k$, equality constraints for some entries in $\mathcal{I}(k,i)$ have to be introduced. Furthermore, Fig. \ref{fig-sparity-1:b} shows the aggregate sparsity pattern of $\bm{Z}$ within each clique of $\mathcal{K}_b$, still at the block level. The sparsity graph is also chordal and its set of maximal cliques is given by
{\small
\begin{equation}\label{eq-sparsity-2}
    \begin{aligned}
        \!\!\! \mathcal{K}_b^{ki} \!\!=\!\! & \left\{\!\mathcal{C}_{bj}|j\!\!=\!\!1,2,3, \mathcal{C}_{b1}\!\!=\!\! \{\mathcal{I}(M\mathbbm{1},\!\cdot),\! \mathcal{I}(D\mathbbm{1},\!\cdot),\! \mathcal{I}(\bm{\omega}_i^k,\!\cdot),\! \mathcal{I}(-1)\}, \right. \!\!\!\! \\
        \!\!\!~&\left. \mathcal{C}_{b2} \!\!=\!\!  \{\!\mathcal{I}(\bm{l}_{\text{m}i}^k,\!\cdot), \mathcal{I}(\bm{l}_{\text{d}i}^k,\!\cdot),\! \mathcal{I}(-1)\},\! \mathcal{C}_{b3} \!\!=\!\! \{\mathcal{I}(\bm{\theta}_i^k, \cdot), \mathcal{I}(-1) \} \right\} \!\!\!\!
    \end{aligned}
\end{equation}
}

Then we investigate aggregate sparsity patterns at the element level for each clique in $\mathcal{K}_b^{ki}$ individually. For cliques $\mathcal{C}_{b1}$ and $\mathcal{C}_{b2}$, they have an analogous aggregate sparsity pattern at the element level. Taking a power grid with 3 generators and $n_c\!\!=\!\!1$ as an example, Fig. \ref{fig-sparity-2:a} shows the aggregate sparsity pattern within $\mathcal{C}_{b1}$ or $\mathcal{C}_{b2}$. Nodes with the same color are associated with indices corresponding to the same generator. It can be found that the sparsity graph is chordal, and each set of nodes corresponding to the same generator and node $\mathcal{I}(-1)$ consist of a maximal clique. Thus sets of maximal cliques for $\mathcal{C}_{b1}$ and $\mathcal{C}_{b2}$, i.e., $\mathcal{K}_{e1}^{ki}$ and $\mathcal{K}_{e2}^{ki}$, are given by (\ref{eq-2-1-ap-1:1}) and (\ref{eq-2-1-ap-1:2}), respectively.

For clique $\mathcal{C}_{b3}$, Fig. \ref{fig-sparity-2:b} gives an example of its aggregate sparsity pattern, where the power grid consists of 4 buses and $n_c \!\!=\!\! 2$. Nodes with the same color are associated with indices corresponding to the same collocation point. The induced subgraph for any set of nodes with the same color is the same as the underlying graph of the power grid. With numbers in colored nodes representing bus numbers of power grid, the underlying graph of power grid contains two maximal cliques, i.e., $\{\!1,2,3\!\}$ and $\{\!3,4\!\}$. Node $\mathcal{I}(-1)$ and all nodes that correspond to each maximal clique of the underlying graph of power grid, form a maximal clique of the sparsity graph of $\mathcal{C}_{b3}$. These two maximal cliques are shown in Fig. \ref{fig-sparity-2:b} as the two groups of nodes linked by blue edges and red edges, respectively. We can find that the aggregate sparsity pattern within $\mathcal{C}_{b3}$ is fully determined by the topology of power grids, and the sparsity graph of $\mathcal{C}_{b3}$ is chordal if and only if $\mathcal{G}_n$ is chordal. With the assumption that $\mathcal{G}_n$ is chordal, we can always find a chordal extension for $\mathcal{G}_n$. Thus the set of maximal cliques for $\mathcal{C}_{b3}$, i.e., $\mathcal{K}_{e3}^{ki}$, is given by (\ref{eq-2-1-ap-1:3}).
\begin{figure}[t]
    \centering
    \subfigure[]{\label{fig-sparity-1:a} \includegraphics[width=4cm]{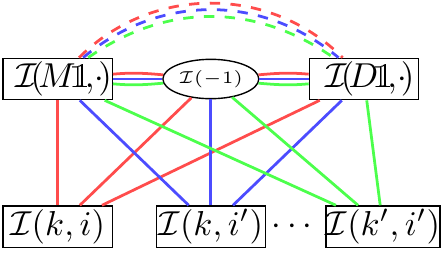}}
    \subfigure[]{\label{fig-sparity-1:b} \includegraphics[width=3.5cm]{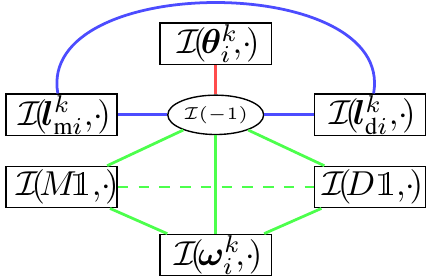}}
    \caption{Aggregate sparsity pattern of matrix $\bm{Z}$ at the block level.}
    \label{fig-sparity-1}
\end{figure}
\begin{figure}[t]
    \centering
    \subfigure[]{\label{fig-sparity-2:a} \includegraphics[width=4cm]{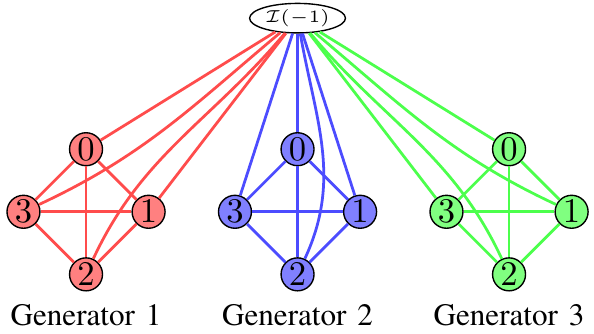}}
    \subfigure[]{\label{fig-sparity-2:b} \includegraphics[width=3.5cm]{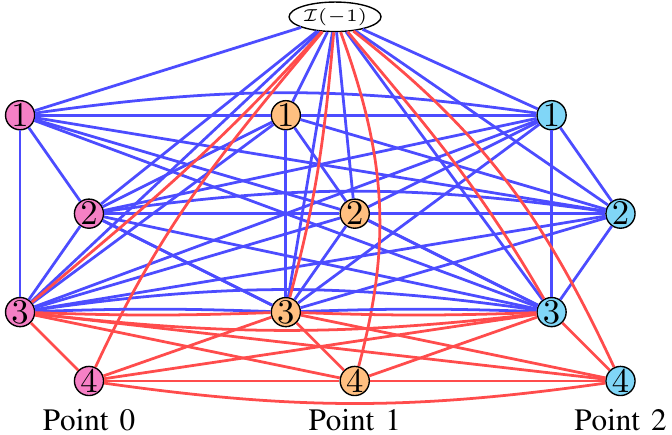}}
    \caption{Aggregate sparsity pattern of matrix $\bm{Z}$ at the element level.}
    \label{fig-sparity-2}
\end{figure}
\end{proof}

\begin{corollary}\label{corollary-1}
    Let $\mathcal{G}(\mathcal{V}, \mathcal{E})$ be a chordal graph and let $\mathcal{K}= \{\mathcal{C}_1,...,\mathcal{C}_{n_{mc}}\}$ be the set of all maximal cliques. Then $\bm{Z}' \in \mathbb{S}^{|\mathcal{V}|}_{\mathcal{E}}$ has a rank-1 positive semi-definite completion if and only if $\bm{Z}'$ satisfies that $\forall i \in \{1,...,n_{mc}\}$, $\mathcal{S}_{\mathcal{C}_i}(\bm{Z}') \succeq 0$ and $\max\{ \rank(\mathcal{S}_{\mathcal{C}_i}(\bm{Z}'))| i \in \{1,...,n_{mc}\} \} = 1$.
\end{corollary}
\begin{proof}
    Corollary \ref{corollary-1} is a corollary of \cite[Therome 2.5]{4-503} and \cite[Theorem 1.5]{4-504}.
\end{proof}

\begin{proposition}\label{pro-2-1-added}
    Equation (\ref{eq-6-36:2}) is equivalent to
    \begin{equation}\label{eq-6-38}
        \begin{aligned}
             \mathcal{Y}_s^{(\kappa + 1)} \!\!= &  \left[\!
            \argmin_{ \rank( \bm{Y}_{\mathcal{C}_{ej}} \!) \leq 1 } \!\!\Vert  \bm{Y}_{\mathcal{C}_{ej}} \!\!-\!\! (\! \hat{\bm{Z}}_{\mathcal{C}_{ej}}^{(\kappa + 1)}  \!\!+\!\! \frac{1}{\tilde{\rho}}  \hat{\bm{\Lambda}}_{\mathcal{C}_{ej}}^{(\kappa)}  )  \! \Vert_F^2
            \!\right]^T \\
            & ~ \!\text{with}~ \mathcal{C}_{ej} \in \mathcal{K}_s  ~~~\forall s \in \mathbb{P}
        \end{aligned}
    \end{equation}
\end{proposition}
\begin{proof}
In (\ref{eq-6-36:2}), the objective function is given by
\begin{equation}\label{eq-6-37}\nonumber
    \begin{aligned}
        & L_s(\hat{\mathcal{Z}}_s^{(\kappa + 1)}\!,\! \mathcal{Z}_{\text{a}}^{(\kappa)}\!,\! \mathcal{Y}_s,\! \mathcal{A}_s^{(\kappa)}\!,\! \tilde{\mathcal{A}}_s^{(\kappa)} ) \!=\! \underbrace{ \frac{\tilde{\rho}}{2} \Vert \mathcal{Y}_s \!-\! \hat{\mathcal{Z}}_s^{(\kappa + 1)} \!-\! \frac{1}{\tilde{\rho}} \tilde{\mathcal{A}}_s^{(\kappa)}  \Vert_F^2}_{\mathcal{Y}_s~\text{involved}}  \\
        &\!\!-\! \frac{1}{2 \tilde{\rho}}  \Vert\!  \tilde{\mathcal{A}}_s^{(\kappa)}  \!\Vert 
        \!\!+\!\!  \Tr\!\!\left(\! \tilde{\mathcal{A}}_s^{(\kappa)T^*} \!\!\!\!\circ\! \hat{\mathcal{Z}}_s^{(\kappa + 1)}  \!\!\right)  \!\!+\!\! \frac{\rho}{2} \Vert \!\zeta_s^A\!(\!\hat{\mathcal{Z}}_s^{(\kappa + 1)}\!) \!\!-\!  \zeta_s^B\!(\! \mathcal{Z}_{\text{a}}^{(\kappa)} \!) \Vert_F^2 \\
        &  \!+\!   \Tr\left( \mathcal{A}_s^{(\kappa)T^*} \circ \left(\zeta_s^A(\hat{\mathcal{Z}}_s^{(\kappa + 1)}  ) \!-\!  \zeta_s^B( \mathcal{Z}_{\text{a}}^{(\kappa)} ) \right) \right) + \hat{J}_{s}( \hat{\mathcal{Z}}_s^{(\kappa + 1)} ) 
    \end{aligned}
\end{equation}
and $\mathcal{Y}_s$ is involved only in the first term. Dropping other terms and $\tilde{\rho}/{2}$ in the first term results in (\ref{eq-6-38}).
\end{proof}

\subsection{Parameter Settings}\label{appendix_5}

\begin{table}[h]
	\caption{Parameters of generators.}
    \label{tb_appendix_4_1}
    {\footnotesize{
	\begin{tabular*}{\hsize}{@{}@{\extracolsep{\fill}}lllllll@{}}
	\toprule
     Generator & $\underline{m}_i$ & $\overline{m}_i$ & $\overline{d}_i$ & $\underline{d}_i$  \\
	\midrule
	$\mathcal{N}_{v_o}$ and $\mathcal{N}_g$  & $0.5 \widetilde{m}_i$ &  $0.5 \widetilde{m}_i$ & $0.5 \widetilde{d}_i$ &  $0.5 \widetilde{d}_i$  \\
	$\mathcal{N}_{v_d}$ (low inertia)        & $0.01 \widetilde{m}_i $ &  $0.01 \overline{m}_i$ & $0.01 \widetilde{d}_i$ &  $ \widetilde{d}_i$   \\
    $\mathcal{N}_{v_d}$ (high inertia)       & $0.5 \widetilde{m}_i$ &  $0.5 \widetilde{m}_i$ &  $0.01 \widetilde{d}_i$ &  $ \widetilde{d}_i$   \\
    $\mathcal{N}_{v_m}\!$                    & $0.01 \widetilde{m}_i$ &  $ \widetilde{m}_i$  &  $0.5 \widetilde{d}_i$ &  $ 0.5 \widetilde{d}_i$  \\
    $\!\mathcal{N}_{v_{dm}}$                 & $0.01 \widetilde{m}_i$ &  $ \widetilde{m}_i$  &  $0.01 \widetilde{d}_i$ &  $ \widetilde{d}_i$  \\
	\bottomrule
    \end{tabular*}
    }}
    \small{ Note: $ \widetilde{m}_i \!=\! \frac{ 10 p_{i,max}}{\omega_{syn}}$ with $p_{i,max}$ being maximal steady-state active power output of generator $i$ and $\omega_{syn}$ being the synchronous angular speed. $ \widetilde{d}_i \!=\! \frac{ 2 p_{i,max}}{2 \pi }$. $\overline{p}_g \!=\! -\underline{p}_g = 3 p_{i,max}$ for each generator.}
\end{table}

\begin{table}[h]
    \caption{Frequency bounds.}
    \vspace{-2pt}
    \label{tb_appendix_4_2}
    {\footnotesize{
	\begin{tabular*}{\hsize}{@{}@{\extracolsep{\fill}}llll@{}}
	\toprule
	 Disturbance & $\!\!\!\!$Time interval& $\!\!\!\!\underline{\omega}^k(t)$ & $\!\!\!\!\overline{\omega}^k(t)$  \\
	\midrule
	$\mathcal{D}_1$ and $\mathcal{D}_2$ & $\!\!\!\!\{\![0s,\! 15s),\! [15s, \!30s]\!\}$     & $\!\!\!\!\{\!49.5,49.85\}$   & $\!\!\!\!\{\!50.5, 50.15\!\}$   \\
	$\mathcal{D}_3$        & $\!\!\!\![0s, 30s]$     & $\!\!\!\!49.85$  & $\!\!\!\!50.15$  \\
	$\mathcal{D}_4$    & $\!\!\!\!\{\![0s, 15s),\! [15s,\! 30s]\!\}$     & $\!\!\!\!\{\!49, 49.5\}$      & $\!\!\!\!\{\!51, 50.5\!\}$ \\
	\bottomrule
    \end{tabular*}
    }}
\end{table}

\begin{table}[t]
	\caption{Parameters of disturbances.}
    \label{tb_appendix_4_3}
    {\footnotesize{
	\begin{tabular*}{\hsize}{@{}@{\extracolsep{\fill}}lllll@{}}
	\toprule
	 Test system &  $\mathcal{D}_1$ &  $\mathcal{D}_2$ &  $\mathcal{D}_3$ &  $\mathcal{D}_4$\\
	\midrule
    AU14Gen         & 203   & 508 & 404 & (212, 217) \\
    IEEE 14-bus     & 2   & 9   & 6  & (9, 14)    \\
    IEEE 39-bus     & 32  & 8   & 39 & (17, 27)   \\
    IEEE 118-bus    & 25  & 54  & 89 & (43, 44)   \\
    ACTIVSg200      & 127 & 100 & 155 & (177, 58) \\
	\bottomrule
    \end{tabular*}
    }}
    \small{ Note: The above table shows location of disturbances, with numbers denoting bus number. Test systems are at the equilibrium point at $t = t_0^{-}$, each disturbance occurs at $t=t_0$, and $P_0$ denotes the initial load power or generation power. For disturbances in $\mathcal{D}_1$, step amplitude is set to $-50\%P_0$, where $P_0$ denotes the initial load power for load buses or generation power for generator buses. For disturbances in $\mathcal{D}_2$, height of ramp and duration of ramp are set to  $-50\%P_0$ and 5 s, respectively. Disturbances in $\mathcal{D}_3$ are emulated by a random power disturbance which changes its value randomly at a equal interval being 0.5 s according to a uniform distribution with the interval being $[-20\%P_0, +20\%P_0]$. Disturbances in $\mathcal{D}_4$ are assumed occurring at the middle of the branch, with short circuit resistance being 0, and being cleared by disconnecting the two sides breakers of the branch after 0.1 s. $\rho_k$ is set to $0.15$, $0.15$, $0.6$ and $0.1$ for disturbances in $\mathcal{D}_1$, $\mathcal{D}_2$, $\mathcal{D}_3$ and $\mathcal{D}_4$, respectively.}
\end{table}

\end{document}